\documentclass{article}

\usepackage{graphicx}%
\usepackage{multirow}%
\usepackage{amsmath,amssymb,amsfonts, amsthm,amssymb}%
\usepackage{amsthm}%
\usepackage{mathrsfs}%
\usepackage[title]{appendix}%
\usepackage{xcolor}%
\usepackage{textcomp}%
\usepackage{manyfoot}%
\usepackage{booktabs}%
\usepackage{algorithm}%
\usepackage{algorithmicx}%
\usepackage{algpseudocode}%
\usepackage{listings}%
\usepackage{comment}
\usepackage{amsmath}
\usepackage{mathtools}

\usepackage{url}
\usepackage{mathtools}  
\newtheorem{prop}{Proposition}
\usepackage[english]{babel}
\usepackage[letterpaper,top=2cm,bottom=2cm,left=3cm,right=3cm,marginparwidth=1.75cm]{geometry}
\usepackage[colorlinks=true, allcolors=blue]{hyperref}

\title{Killing tensors of Weyl's Class}
\author{Dionysios Kokkinos \footnote{Department of Mechanical Engineering, Hellenic Mediterranean University, Heraklion, Crete, Greece \\  kokkinos@physics.uoc.gr}}

\begin{document}
\maketitle







\abstract{This work represents an initial step towards a systematic framework for identifying hidden symmetries in algebraically general (Petrov type I) spacetimes. Motivated by this, we focus on Weyl's class, which provides a simple yet a class of algebraically general solutions of static, axisymmetric geometries. Employing the canonical forms of rank-2 Killing tensors, we systematically derive the constraints that must be satisfied for these spacetimes to admit genuine hidden symmetries. Our analysis yields an analytical characterization of the obstructions to integrability within subclasses of Weyl's class, both in vacuum and electrovacuum. In particular, we find that no two-Killing-vector member of Weyl's class, in vacuum or electrovacuum, admits an irreducible Killing tensor capable of supplying the additional integral of motion required for complete integrability. Members of Weyl's class admitting a third additional Killing-vector isometry (i.e. Levi-Civita) are already completely integrable through their manifest symmetries alone admitting only reducible Killing tensors. This result provides an analytic complement to previous analytical and numerical investigations, and establishes a systematic connection between the absence of hidden symmetries and the breakdown of complete integrability. Beyond the specific classification obtained here, these results demonstrate the potential of our approach regarding the canonical forms of Killing tensor as a systematic tool for probing hidden symmetries in Petrov type I spacetimes.
}

\section{Introduction}

Einstein's theory of gravity still stands as the gravitational \textit{paradigm} of our time, serving simultaneously as the theoretical anchor for later gravitational theories. Although this theory is often considered thoroughly explored, one cannot disregard the possibility that certain classes of solutions remain to be investigated, leaving open the prospect that rich phenomenology could emerge from their analysis. In particular, Petrov type I solutions remain largely unexplored, since a systematic method for their extraction is still absent.

Two of the most known and general classes of solutions, which are of type I, can be obtained by assuming axial symmetry and stationarity or staticity resulting in Weyl-Lewis-Papapetrou's class (WLP) or Weyl's class correspondingly. These two quite known classes of spacetimes initially were extracted in vacuum limit, but since then, they have been studied in electrovacuum with (or without) cosmological constant and rarely with a presence of a stress-energy-momentum tensor representing dust, perfect fluid or anisotropic fluids \cite{stephani2009exact, griffiths2009exact, schmidt2008einstein}.  

The main topic of this work is Weyl's class and describes the exterior gravitational field of a static, axially symmetric source in vacuum \cite{weyl1919neue, weyl2012republication}. Algebraically, its Weyl tensor is of type I, distinguishing it from the more restrictive type D solutions (such as Schwarzschild spacetime) that arise as special case. Within this class the most known subclasses are Levi-Civita spacetime \cite{levi1919meccanica}, Zipoy-Voorhees known as $\gamma$-metric \cite{zipoy1966topology,voorhees1970static}  and Chazy-Curzon \cite{chazy1924champ, curzon1925cylindrical} spacetime. 

Despite its age, Weyl's class remains a fertile setting for probing the relationship between the geometric structure of spacetime and the existence of hidden symmetries encoded by higher-rank Killing tensors. In particular, the Zipoy-Voorhees and Chazy-Curzon spacetimes exhibit non-integrable geodesic dynamics, as their geodesic motion lacks the independent constants of motion to ensure complete integrability, thus, the equations of motion cannot be reduced to quadratures and the resulting trajectories may display chaotic behavior \cite{brink2008spacetime, lukes2012nonintegrability, vollmer2017killing, vollmer2015reducibility}. A longstanding debate about the existence of an additional, Carter-like constant exists within the literature. Several works have proposed the presence of such a conserved quantity, while subsequent analytical and numerical investigations have challenged these claims \cite{dubeibe2019geodesic,georgiev2019non}, (see also \cite{maciejewski2013nonexistence} and references therein). Sometimes the lack of a systematic analytical approach gives rise to debates like these and most of the times numerical studies are unable to answer thoroughly. 

The question of whether the geodesic equations of a given spacetime are completely integrable is, in practice, a question about the existence of hidden symmetry enabled by a rank-2 (or higher rank) Killing tensor. In line with this, the central question can be formulated as follows: \textit{if a spacetime with two hypersurface-orthogonal Killing vectors within Weyl's class, in vacuum or electrovacuum, admits completely integrable  geodesic motion, must there exist an irreducible rank-2 (or higher-rank) Killing tensor, part of its  canonical forms \cite{kokkinos2024}, that generates the additional integral of motion required for complete integrability}? \cite{benenti1979remarks, cariglia2014hidden}. Beyond the explicit symmetries generated by Killing vectors, hidden symmetries are encoded in the phase-space structure through irreducible Killing tensors of rank-2 or higher giving rise to additional polynomial integrals of the geodesic flow. The hidden symmetry of Carter's family known as Carter's constant \cite{Carter1968b}, generated by an irreducible rank-2 Killing tensor tied to the D type of the Weyl tensor \cite{walker1970quadratic}, rendering the geodesic Hamilton-Jacobi equation separable. For the static, axisymmetric vacuum solutions of Weyl's class this is not the case though. Beyond the two Killing vectors $\partial_t$ and $\partial_\phi$ guaranteed by the initial ansatz, nothing in the construction of Weyl's class seems to generate a fourth integral of motion. Consequently, for spacetimes admitting two independent Killing vectors, determining the existence or absence of irreducible Killing tensors provides a geometric criterion for assessing whether the remaining integrals required for complete integrability can arise from hidden symmetries. The present work addresses a specific question within this setting: \emph{can irreducible Killing tensors, and hence a genuine hidden symmetry associated with integrability, exists within the broader class of Weyl spacetimes in electrovacuum and vacuum limit?}

Recently, a sophisticated method to determine whether a spacetime admits irreducible Killing tensors was developed in \cite{gray2025lower}. The authors employ a metric ansatz that decomposes the $d$-dimensional manifold into two special transverse coordinates $t,r$ and a $(d-2)$-dimensional base space, with a shift vector $\nu^A$ mixing $t$ with the base coordinates. The key ingredient for irreducibility is the non-commutativity of the base isometry algebra: the shift vector demotes non-commuting base symmetries from manifest to hidden, while their Casimir combination lifts to an irreducible rank-2 Killing tensor of $\mathcal{M}$. A vanishing shift vector, by contrast, leaves every base isometry manifest, and hence reducible. It therefore remains to be determined whether Weyl's class and by extension any static, hypersurface-orthogonal axisymmetric metric with no $g_{t\varphi}$ cross term can be cast in this ansatz at all. If it can, the resulting vanishing shift vector would trivially imply only reducible Killing tensors, in line with our results. If it cannot, this would oblige us to pursue other methods to answer the question already posed.

Our assessment is that a sufficient answer may be given by reducing the Killing equations of the canonical forms of Killing tensor to explicit differential conditions on the metric potentials, and confronting these conditions with the generic multipole expansion of the vacuum potentials. This is valid because pursuing to extract algebraically general solutions, subclasses of Weyl's class, only the assumption of existence of a Killing tensor works in contrast to the Killing–Yano tensor since its existence \textbf{constrains} the algebraic character of our solution forbidding algebraically general solutions \cite{papakostas1985space, collinson1971special}. On the other hand, this approach does not address the possibility of hidden symmetries generated by Killing tensors of rank higher than two, nor does it establish their reducibility in the general case. Such higher rank Killing tensors may provide additional integrals of motion even when no rank-2 Killing tensor is present, as discussed, for example, in \cite{vollmer2015reducibility}. Consequently, a complete assessment of integrability would require extending the analysis to higher-rank Killing tensors.

This analytical approach provides a systematic basis, and Weyl's class serves as a \textit{crash test} for identifying irreducible rank-2 Killing tensors in type I spacetimes, investigating their possible existence within subclasses of Weyl's class. It is worth noting that the existence of such an irreducible Killing tensor is a sufficient condition for complete integrability of the geodesic flow, since the quadratic quantity $Q=K^{\mu\nu}p_\mu p_\nu$ is automatically conserved and served as a fourth constant of motion along with the Hamiltonian and with the linear invariants $E,L$ generated by the two Killing vectors. In this scheme, Hamilton-Jacobi separation is not need to be exhibited for this conclusion to hold. Separability is a stronger and more constructive condition, since it additionally places the conserved quantities via quadrature. In this fashion, it sets up a framework to investigate integrability apart from the geodesic equations themselves, offering an answer to the question posed, from a different point of view. At the same time, with this line of work we give an answer about the integrability properties of geodesic motion in the whole class, incorporating the Zipoy-Voorhees ($\gamma$-metric) and Curzon-Chazy spacetimes.


This paper starts by introducing the necessary notation of Newman-Penrose formalism in \autoref{Chapter 2}. In \autoref{Chapter 3} we describe the Weyl's class in electrovacuum by considering a null tetrad frame presenting the corresponding spin coefficients and the field equations in this limit. The main analysis is made in vacuum \autoref{Chapter 4} while in \autoref{Chapter 5} we discuss unique subclasses of Weyl's class such as Zipoy-Voorhees spacetime ($\gamma$-metric) and Chazy-Curzon spacetime and how the integrability of Schwarzchild spacetime is adapted to our setup. Afterwards, the corresponding analysis in electrovacuum takes place in \autoref{Chapter 6} respectively. At last, a summary is given at the last chapter of this paper in \autoref{Chapter 8}.

\section{Newman-Penrose Formalism}\label{Chapter 2}

The Newman-Penrose Formalism is a widely known formalism that was presented by Newman and Penrose \cite{newman1962approach} and was analyzed geometrically by Cahen, Debever and Defrise \cite{cahen1967complex}, \cite{debeverriemann}.

The main concept of the formalism could be briefly described as follows: \textbf{The need to interpret the gravitational radiation more conveniently forces us to associate the Riemann tensor with isotropic null tetrads}. The latter could happen in a 3-dimensional complex bivector space ($C_3$) spanned by self-dual 2-forms (bivectors). The metric can be put in the form

\begin{equation}ds^2 = 2(\boldsymbol{\theta}^1 \boldsymbol{\theta}^2 - \boldsymbol{\theta}^3 \boldsymbol{\theta}^4) ,\end{equation}
where the general metric $g_{\mu \nu}$ is the following and equal to its inverse $g^{\mu \nu}$

\begin {equation}g_{\mu \nu} = l_\mu n_\nu + n_\mu l_\nu   - m_\mu \bar{m}_\nu - \bar{m}_\mu m_\nu  = \begin{pmatrix}
0 &1&0&0\\
1&0&0&0\\
0&0&0&-1\\
0&0&-1&0
\end{pmatrix}.\end{equation}
The pseudo-orthonormal basis contains two real and two complex conjugate vectors

\begin{equation} \boldsymbol{\theta}^1 \equiv n_\mu dx^\mu, \hspace{0.8cm} \boldsymbol{\theta}^2 \equiv l_\mu dx^\mu, \hspace{0.8cm} \boldsymbol{\theta}^3 \equiv - \bar{m}_\mu dx^\mu, \hspace{0.8cm} \boldsymbol{\theta}^4 \equiv - m_\mu dx^\mu,  \end{equation}
the non-zero orthogonality properties of the vector components are

\begin{equation} l_\mu n^\mu = 1 = - m_\mu \bar{m}^\mu .\end{equation}
The directional derivatives (dual basis) of the formalism are given by

$$\boldsymbol{D}  =  l^{\mu} \partial_\mu  ,\hspace{0.8cm}\boldsymbol{\Delta} = n^{\mu} \partial_\mu, \hspace{0.8cm}\boldsymbol{\delta}  = m^{\mu} \partial_\mu, \hspace{0.8cm} \boldsymbol{\bar{ \delta}} = \bar{m}^{\mu} \partial_\mu .$$
The Einstein's Field Equations in this formalism are represented by the corresponding field equations, the Newman-Penrose Field Equations (NPE) (or Ricci identities) which are presented at the electrovacuum limit with the presence of cosmological constant \cite{newman1962approach}.

\small
\begin{equation}\tag{a} D \rho  - \bar{\delta} \kappa = {\rho}^2 + \sigma\bar\sigma+ \rho ( \epsilon + \bar{\epsilon}) - \bar{\kappa} \tau - \kappa \left[2(\alpha +\bar{\beta}) + (\alpha - \bar{\beta}) - \pi\right] + \Phi_{00},\end{equation}
\begin{equation}\tag{b} \delta \kappa -D\sigma = -\sigma(\rho+\bar\rho+3\epsilon-\bar\epsilon)+ \kappa \left[ \tau - \bar{\pi} +2(\bar{\alpha} +\beta) - (\bar{\alpha} - \beta) \right] - \Psi _0 , \end{equation}
\begin{equation}\tag{c} D\tau = \Delta \kappa + \sigma(\pi+\bar\tau) + \rho(\bar{\pi}+\tau)+ \tau(\epsilon - \bar{\epsilon})    -2\kappa \gamma - \kappa (\gamma + \bar{\gamma}) + \Psi_1 +\Phi_{01} ,\end{equation}
\begin{equation}\tag{i}\label{i} D\nu - \Delta \pi = \lambda(\bar\pi+\tau)+\mu(\pi + \bar{\tau}) +\pi(\gamma - \bar{\gamma}) -2\nu \epsilon - \nu (\epsilon + \bar{\epsilon}) +\Psi_3+\Phi_{21}, \end{equation}
\begin{equation}\tag{g} \bar{\delta} \pi -D\lambda= - \pi(\pi + \alpha - \bar{\beta}) -\bar\sigma\mu -\lambda\rho + \nu\bar{\kappa} +\lambda(3\epsilon-\bar\epsilon) -\Phi_{20},\end{equation}
\begin{equation}\tag{p}\label{p} \delta \tau -\Delta\sigma= \tau (\tau - \bar{\alpha} + \beta)+\sigma\mu+\bar\lambda\rho -\bar{\nu} \kappa  -\sigma(3\gamma-\bar\gamma)+\Phi_{02},\end{equation}
\begin{equation}\tag{h} D\mu - \delta\pi = \mu \bar{\rho} +\sigma\lambda+ \pi(\bar{\pi} - \bar{\alpha} +\beta) -\mu (\epsilon +\bar{\epsilon}) - \kappa \nu + \Psi_2 + 2\Lambda,\end{equation}
\begin{equation}\tag{n} \delta\nu - \Delta \mu = \mu (\mu + \gamma + \bar{\gamma}) +\lambda\bar\lambda - \bar{\nu} \pi + \nu (\tau - 2(\bar{\alpha}+\beta) +(\bar{\alpha} -\beta) ) +\Phi_{22},\end{equation}
\begin{equation}\tag{q} \Delta \rho - \bar{\delta} \tau = - \bar{\mu} \rho-\sigma\lambda - \tau(\bar{\tau} + \alpha - \bar{\beta}) + \nu \kappa + \rho(\gamma + \bar{\gamma}) - \Psi_2 - 2\Lambda,\end{equation}
\begin{equation}\tag{k} \delta\rho-\bar\delta\sigma = \rho(\bar{\alpha}+ \beta) -\sigma(3\alpha-\bar\beta) +\tau(\rho-\bar{\rho})+ \kappa(\mu-\bar{\mu}) - \Psi_1+\Phi_{01} ,\end{equation}
\begin{equation}\tag{m}\label{m} \bar{\delta} \mu -\delta\lambda= -\mu (\alpha +\bar{\beta}) -\pi (\mu- \bar{\mu}) - \nu (\rho-\bar{\rho}) -\lambda(3\bar\alpha-\beta)+ \Psi_3-\Phi_{21}, \end{equation}
\begin{equation}\tag{d} D\alpha - \bar{\delta} \epsilon = \alpha(\rho + \bar{\epsilon} -2\epsilon) +\beta\bar\sigma - \bar{\beta}\epsilon-\kappa\lambda - \bar{\kappa}\gamma + \pi (\epsilon + \rho) +\Phi_{10},\end{equation} 
\begin{equation}\tag{e}\label{e} D \beta - \delta{\epsilon} = \sigma(\alpha+\pi)+\beta(\bar{\rho} - \bar{\epsilon}) -\kappa(\mu + \gamma) -\epsilon(\bar{\alpha} - \bar{\pi}) + \Psi_1,\end{equation} 
\begin{equation}\tag{r} \Delta \alpha - \bar{\delta}\gamma = \nu(\epsilon+\rho)-\lambda(\tau+\beta) +\alpha( \bar{\gamma} - \bar{\mu}) +\gamma (\bar{\beta}- \bar{\tau}) - \Psi_3, \end{equation}
\begin{equation}\tag{o}\label{o} -\Delta \beta + \delta \gamma = \gamma(\tau -\bar{\alpha} - \beta) +\mu \tau-\sigma\nu - \epsilon \bar{\nu} - \beta( \gamma - \bar{\gamma} -\mu) +\alpha\bar\lambda+\Phi_{12},\end{equation}
\begin{equation}\tag{l}\label{l} \delta \alpha - \bar{\delta}\beta = \mu \rho -\sigma\lambda +\alpha (\bar{\alpha} - \beta) - \beta(\alpha - \bar{\beta})+  \gamma(\rho - \bar{\rho}) +\epsilon (\mu-\bar{\mu})-\Psi_2  + \Lambda +\Phi_{11},\end{equation}
\begin{equation}\tag{f}\label{f} D\gamma - \Delta \epsilon = \alpha(\tau + \bar{\pi}) + \beta( \bar{\tau} + \pi) - \gamma ( \epsilon+\bar{\epsilon}) - \epsilon (\gamma +\bar{\gamma})   + \Psi_2 - \Lambda - \kappa \nu +\tau \pi +\Phi_{11},\end{equation}
\begin{equation}\tag{j} \bar{\delta}\nu -\Delta\lambda=\lambda(\mu+\bar\mu+3\gamma-\bar\gamma) - \nu\left[2(\alpha +\bar{\beta} ) + (\alpha - \bar{\beta})  + \pi - \bar{\tau}  \right] + \Psi_4.\end{equation}
\normalsize
In this formalism, the 10 Weyl's components are represented by the 5 complex scalar functions as follows
\small
$$ \Psi_0 = C_{\kappa \lambda \mu \nu} l^\kappa m^\lambda l^\mu m^\nu ,$$
$$ \Psi_1 = C_{\kappa \lambda \mu \nu} l^\kappa n^\lambda l^\mu m^\nu  ,$$
\begin{equation} \Psi_2 = \frac{1}{2}C_{\kappa \lambda \mu \nu} l^\kappa n^\lambda \left[ l^\mu n^\nu -    m^\mu \bar{m}^\nu  \right] , \end{equation}
$$ \Psi_3 = C_{\kappa \lambda \mu \nu} n^\kappa l^\lambda n^\mu \bar{m}^\nu, $$
$$ \Psi_4 = C_{\kappa \lambda \mu \nu} n^\kappa \bar{m}^\lambda n^\mu \bar{m}^\nu .$$
\normalsize
The Maxwell equations follow
\begin{equation}\label{eq:Maxwell1}
    D\Phi_1 - \bar\delta\Phi_0 = (\pi-2\alpha)\Phi_0 +2\rho\Phi_1 -\kappa \Phi_2
,\end{equation}
\begin{equation}\label{eq:Maxwell2}
    D\Phi_2-\bar\delta\Phi_1 = -\lambda\Phi_0 +2\pi\Phi_1+(\rho-2\epsilon)\Phi_2
,\end{equation}
\begin{equation}\label{eq:Maxwell3}
    -\Delta\Phi_0+\delta\Phi_1 = (\mu-2\gamma)\Phi_0+2\tau\Phi_1-\sigma\Phi_2
 ,\end{equation}
\begin{equation}\label{eq:Maxwell4}
    -\Delta\Phi_1+\delta\Phi_2 = -\nu\Phi_0+2\mu\Phi_1+(\tau-2\beta)\Phi_2
.\end{equation}
\normalsize
where the electromagnetic scalars are defined by the relations below
\begin{equation}
    \Phi_{ab} = \kappa_0 \Phi_a \Phi_b~;~~~~~~a,b= 0,1,2, 
\end{equation}
\begin{equation}
\Phi_0 \equiv F_{\mu\nu}\,l^\mu m^\nu,
\end{equation}
\begin{equation}
\Phi_1 \equiv \frac{1}{2}F_{\mu\nu}\big(l^\mu n^\nu + \bar m^\mu m^\nu\big),
\end{equation}
\begin{equation}
\Phi_2 \equiv F_{\mu\nu}\,\bar m^\mu n^\nu.
\end{equation}
In this point we ought to note that in a static axisymmetric spacetime should exist an electromagnetic potential $A(\rho,z)=(V(\rho,z),0,0,0)$ which is able to give rise to the Maxwell tensor $F_{\mu\nu}\equiv \partial_\mu A_{\nu} - \partial_\nu A_{\mu}$. As a matter of fact it is easy for someone to show that with this potential the only non-zero electromagnetic scalars are $\Phi_0, \Phi_2$ which satisfy the following relation
\begin{equation}
    \Phi_0=-\bar\Phi_2.
\end{equation}

\section{Weyl's class}\label{Chapter 3}

This class of spacetimes has been found a century ago by Weyl and describes the gravitational field of a static and axially symmetric gravitational source in vacuum \cite{weyl1919neue, weyl2012republication}. Algebraically its Weyl tensor is general characterizing this class as type I according to Petrov classification. 

\begin{equation}
    ds^2 = e^{2\lambda} dt^2 - e^{-2\lambda} \left[  \rho^2 d\phi^2 +e^{2\mu} (d\rho^2 + dz^2) \right].
\end{equation}
The metric is presented in the Weyl's canonical coordinates and both functions $\lambda$ and $\mu$ depend on the non-ignorable coordinates $\rho \in [0,+\infty) ,z \in (-\infty, +\infty)$ while the ignorable one are $t\in(-\infty, +\infty)$ and $\phi\in[0, 2\pi)$. 
Function $\lambda(\rho,z)$ represents the gravitational potential generated by the source and its general form can be obtained by solving the Laplacian. At the same time, the following field equations characterize fully our spacetime in vacuum

\begin{equation}\label{Laplacian}
    \nabla^2 \lambda =0 ~~\Rightarrow \qquad \lambda_{\rho \rho} +\frac{\lambda_\rho}{\rho} + \lambda_{zz} = 0,
\end{equation}
\begin{equation}\label{mu_rho}
\mu_\rho = \rho(\lambda^2_\rho - \lambda^2_z) ,
\end{equation}
\begin{equation}\label{mu_z}
    \mu_z = 2\rho\lambda_\rho\lambda_z.
\end{equation}

The general solution of this system of equations in a closed form have been found by Waylen \cite{waylen1982general} while its construction to dimensions $D\ge4$ made by Emparan and Reall \cite{emparan2002generalized}.

In pursuit of type I spacetimes, subclasses of Weyl's class, which will consist possible candidates for accommodating hidden symmetries we scope to investigate their existence not only in vacuum but also in electrovacuum limit. For this reason after the definition of null tetrad frame and the presentation of the spin coefficient we will give the NP equations with the presence of Maxwell tensor. Thus, the field equations will be presented implied by a static potential which yields the Maxwell equations' constraint, namely, $\Phi_0+\bar\Phi_2=\Phi_1=0$.

\begin{equation}
    n_\mu = (\frac{e^\lambda}{\sqrt2}, \frac{\rho e^{-\lambda}}{\sqrt2}, 0 ,0),
\end{equation}
\begin{equation}
    l_\mu = (\frac{e^\lambda}{\sqrt2}, -\frac{ \rho e^{-\lambda}}{\sqrt2}, 0 ,0),
\end{equation}
\begin{equation}
    m_\mu = (0,0,\frac{e^{\mu-\lambda}}{\sqrt2}, i\frac{e^{\mu - \lambda}}{\sqrt2}),
\end{equation}
\begin{equation}
    m_\mu = (0,0,\frac{e^{\mu-\lambda}}{\sqrt2}, -i\frac{e^{\mu - \lambda}}{\sqrt2}).
\end{equation}
The directional derivatives take the expression as follows

\begin{equation}
    D = \frac{l_\phi\partial_t - l_t\partial_\phi}{n_tl_\phi-l_tn_\phi}= \frac{\rho e^{-\lambda} \partial_t -e^\lambda \partial_\phi}{\sqrt2\rho},
\end{equation}
\begin{equation}
    \Delta = - \frac{n_\phi\partial_t - n_t\partial_\phi}{n_tl_\phi-l_tn_\phi} = \frac{\rho e^{-\lambda} \partial_t +e^\lambda \partial_\phi}{\sqrt2\rho},
\end{equation}
\begin{equation}
    \delta = \frac{m_z\partial_\rho - m_\rho\partial_z}{m_\rho \bar{m}_z-m_z \bar{m}_\rho} = - \frac{e^{\lambda-\mu}}{\sqrt2}(\partial_\rho+i\partial_z),
\end{equation}
\begin{equation}
    \bar\delta = - \frac{\bar{m}_z\partial_\rho - \bar{m}_\rho\partial_z}{m_\rho \bar{m}_z-m_z \bar{m}_\rho}=- \frac{e^{\lambda-\mu}}{\sqrt2}(\partial_\rho-i\partial_z).
\end{equation}
The Cartan's structure equations are
\begin{equation}d\theta^1 = ( \bar\alpha+\beta-\bar{\pi})\theta^1\wedge\theta^3 + (\alpha+\bar\beta-\pi) \theta^1\wedge\theta^4 -\bar{\nu}\theta^2\wedge\theta^3 -\nu\theta^2\wedge\theta^4,\end{equation}
\begin{equation}d\theta^2 = \kappa\theta^1\wedge\theta^3 +\bar{\kappa} \theta^1\wedge\theta^4 -(\bar\alpha+\beta-\tau)\theta^2\wedge\theta^3 -(\alpha+\bar\beta-\bar{\tau})\theta^2\wedge\theta^4,\end{equation}
\begin{equation}d\theta^3 =(\alpha-\bar{\beta})\theta^3\wedge\theta^4,\end{equation}
\begin{equation}d\theta^4 =  -(\bar{\alpha}-\beta)\theta^3\wedge\theta^4,\end{equation}
where the spin coefficients take the following forms in respect to tetrads' components

\begin{equation}
    a+\bar\beta-\pi = \frac{-l_\phi\bar\delta n_t +l_t \bar\delta n_\phi}{n_tl_\phi-l_tn_\phi},
\end{equation}
\begin{equation}
   \nu = \frac{-n_\phi\bar\delta n_t +n_t \bar\delta n_\phi}{n_tl_\phi-l_tn_\phi},
\end{equation}
\begin{equation}
    \bar\kappa = \frac{-l_\phi\bar\delta l_t +l_t \bar\delta l_\phi}{n_tl_\phi-l_tn_\phi},
\end{equation}
\begin{equation}
    \alpha+\bar\beta-\bar\tau = \frac{-n_\phi\bar\delta l_t +n_t \bar\delta l_\phi}{n_tl_\phi-l_tn_\phi},
\end{equation}
\begin{equation}
    \alpha-\bar\beta = \frac{\bar m_z \delta \bar m_\rho - \bar m_\rho \delta \bar m_z +m_z \bar\delta \bar m_\rho - m_\rho \bar\delta \bar m_z}{m_\rho \bar m_z - \bar m_\rho m_z}.
\end{equation}

Consequently the spin coefficients relations can be expressed as follows 

\begin{equation}\label{spin}
    \sigma=\lambda = \gamma=\epsilon = \mu = \rho=\pi-\bar\pi=0,
\end{equation}
\begin{equation}
    \bar{\alpha}+\beta = \bar{\pi} +\tau = \bar\kappa +\nu = 0,
\end{equation}

\begin{equation}\label{pi}
    \pi = \frac{1}{2}\frac{\bar\delta \rho}{\rho} = -\frac{e^{\lambda-\mu}}{2\sqrt{2}\rho},
\end{equation}
\begin{equation}\label{kappa}
    \bar\kappa = -\frac{1}{2}\frac{\bar\delta(\frac{e^{2\lambda}}{\rho})}{\frac{e^{2\lambda}}{\rho}} = \frac{e^{\lambda-\mu}}{\sqrt{2}}\left( \lambda_\rho - \frac{1}{2\rho} -i\lambda_z \right),
\end{equation}
\begin{equation}\label{coefficients}
a = \frac{\bar\delta (\lambda-\mu)}{2} =  \frac{e^{\lambda-\mu}}{2\sqrt{2}}\left[ (\mu-\lambda)_\rho -i(\mu-\lambda)_z \right].    
\end{equation}

At this point, one may observe that Weyl's class admits the so-called symmetric \footnote{In \cite{kokkinos2025}, we incorrectly stated that the extraction of a type I solution could be achieved by employing an antisymmetric null-tetrad transformation, leading to relations such as $\kappa+\bar{\nu}=\pi+\bar{\tau}=\dots=0$. In fact, we employed the symmetric transformation, since these combinations of spin coefficients can be obtained only through the transformation $n\longleftrightarrow -l$, $m\rightarrow -m$.} null-tetrad transformations, also known as prime transformations, which are also employed in the GHP formalism  \cite{geroch1973space,odonnell2003introduction}. Building on our previous work \cite{kokkinos2025}, in which these transformations were employed under the assumption of a sub-form of the Killing tensor with $q=+1$, namely $\lambda_3=0$, one would expect the corresponding Killing equations to admit an explicit solution. 

Moving forward, the NP equations due to relations (\ref{spin})-(\ref{coefficients}) read

\begin{equation}\tag{a)(n}
    \bar\delta \kappa = 2\alpha\kappa - \pi (\kappa+\bar\kappa)-\Phi_{00};  \ \ \  ~\Phi_{00}=\bar\Phi_{22},
\end{equation}
\begin{equation}\tag{b)(j}
    \delta\kappa = -2\kappa(\bar\alpha+\pi) - \Psi_0;  \ \ \  ~\Psi_{0}=\bar\Psi_{4},
\end{equation}
\begin{equation}\tag{g}
    \bar\delta\pi = -\pi(\pi+2\alpha) -\bar\kappa^2-\Phi_{20},
\end{equation}
\begin{equation}\tag{p}
    \delta\pi = -\pi(\pi+2\bar\alpha) -\kappa^2 - \Phi_{02},
\end{equation}
\begin{equation}\tag{l}
    \delta\alpha+\bar\delta\bar\alpha = 4\alpha\bar\alpha - \Psi_2+\Lambda,
\end{equation}
\begin{equation}\tag{f}
    \Psi_2 - \Lambda = \pi^2 - \kappa\bar\kappa,
\end{equation}
\begin{equation}\tag{h)(q}
\delta\pi = -\pi(\pi-2\bar\alpha) -\kappa\bar\kappa -\Psi_2 - 2\Lambda.    
\end{equation}
In terms of the metric functions the NP equations take the following forms
\small
\begin{equation}\tag{a)(n}\label{an}
\lambda_{\rho\rho}+\frac{\lambda_\rho}{\rho}+\lambda_{zz} = e^{-2\lambda}\big(V_\rho^2+V_z^2\big)
\end{equation}

{\footnotesize
\begin{multline}\tag{b)(j}
    \Psi_0 =   \frac{e^{\lambda-\mu}}{2} \biggl[ 2 \left(\lambda_\rho +\frac{1}{2\rho} - \mu_\rho \right) \left( \lambda_\rho -\frac{1}{2\rho} \right) + \lambda_{\rho\rho} +\frac{1}{2\rho^2} -\lambda_{zz} \\
    + i2 \left( \lambda_z\left( \lambda_\rho +\frac{1}{2\rho} - \mu_\rho \right) + \lambda_{z\rho}\right) \biggr] ,
\end{multline}
}
\begin{equation}\tag{g)+(p}
\mu_\rho = \rho(\lambda^2_\rho - \lambda^2_z) -\rho e^{-2\lambda}\big(V_\rho^2-V_z^2\big), 
\end{equation}
\begin{equation}\tag{g)-(p}
        \mu_z = 2\rho \Big[\lambda_\rho\lambda_z -e^{-2\lambda} V_\rho V_z \Big],
\end{equation}
\begin{equation}\tag{l)+(f}
    \Psi_2 - \Lambda = -\frac{e^{2(\lambda-\mu)}}{2} \left[ \lambda^2_\rho - \frac{\lambda_\rho}{\rho} +\lambda^2_z \right],
\end{equation}
\begin{equation}\tag{h)(q}
\Psi_2 +2\Lambda = -\frac{e^{2(\lambda-\mu)}}{2} \left[ \lambda^2_\rho - \frac{\lambda_\rho}{\rho} +\lambda^2_z \right].    
\end{equation}
\normalsize
Even if we allowed the existence of the cosmological constant within our NP equations, the last two equations clarify that the Weyl's class reduce our limit to to the vacuum as expected. Thus,

\begin{equation}\label{lambda}
    \Psi_2 = -\frac{e^{2(\lambda-\mu)}}{2} \left[ \lambda^2_\rho - \frac{\lambda_\rho}{\rho} +\lambda^2_z \right], \hspace{0.5cm} \& \hspace{0.5cm} \Lambda=0.
\end{equation}
In parallel, the contribution by the Maxwell equations \eqref{eq:Maxwell1}-\eqref{eq:Maxwell4} results in the following relations
\begin{equation}\label{MaxwellEV}
\bar\delta\Phi_0 = (2\alpha-\pi)\Phi_0+\kappa\Phi_2
\quad\Longrightarrow\quad
V_{\rho\rho}+\frac{V_\rho}{\rho}+V_{zz} = 2\big(\lambda_\rho V_\rho+\lambda_z V_z\big)
,\end{equation}
\normalsize
where the scalars take the following forms.

\begin{equation}
\Phi_0=-\frac{e^{-\mu}}{2}\big(F_{t\rho}+iF_{tz}\big),\qquad
\Phi_2=\frac{e^{-\mu}}{2}\big(F_{t\rho}-iF_{tz}\big),
\end{equation}

\section{Killing tensors in Vacuum}\label{Chapter 4}

In a previous study \cite{kokkinos2024}, we made the first algebraic classification of rank-2 Killing tensors in General Relativity, resulting in four distinct families of canonical forms. In the present work, we focus on three of these families, which are algebraically more general than the remaining one and are characterized by the parameter ($q=0,\pm1$), with $q=+1$ corresponding to the most general canonical form. Their compact form is given by
\begin{multline}
K = \lambda_0\big(\theta^1\otimes\theta^1 + q\,\theta^2\otimes\theta^2\big) + \lambda_1\big(\theta^1\otimes\theta^2+\theta^2\otimes\theta^1\big) \\
+ \lambda_2\big(\theta^3\otimes\theta^4+\theta^4\otimes\theta^3\big) + \lambda_3\big(\theta^3\otimes\theta^3+\theta^4\otimes\theta^4\big), \label{eq:Kcanonical}
\end{multline}
The four scalar functions $\lambda_0,\lambda_1,\lambda_2,\lambda_3$ are, a priori, arbitrary functions of $\rho,z$, to be constrained by the Killing equations. Reduction of this equation due to the spin coefficients of Weyl's class yields the following categorization based on the values of $q$:

\scriptsize
\begin{table}[H]
\caption {Killing equations} \label{KillingEqns} 
\begin{center}
\begin{tabular}{c|c|c}
$q = 0$ & $q=-1$ & $q=+1$ \\ \hline
$0= (\kappa+\bar\kappa)(\lambda_1+\lambda_2+\lambda_3)$ &
$0= (\kappa+\bar\kappa)(\lambda_1+\lambda_2+\lambda_3)$ &
$\delta\lambda_0 = 2\left[ \bar\pi\lambda_0 -\kappa (\lambda_1+\lambda_2) - \bar\kappa\lambda_3\right]$\\
$0= (\kappa-\bar\kappa)(\lambda_1+\lambda_2-\lambda_3)$ &
$0= (\kappa-\bar\kappa)(\lambda_1+\lambda_2-\lambda_3)$&
$\delta \lambda_1 = -2\left[ \kappa\lambda_0 -\bar\pi(\lambda_1+\lambda_2) - \pi\lambda_3 \right]$\\
$\delta\lambda_0 = 2\bar\pi\lambda_0$&
$\delta\lambda_0 = 2\bar\pi\lambda_0$&
$\delta \lambda_2 = 4\alpha\lambda_3$\\
$\delta \lambda_1=-\kappa\lambda_0 +2\left[\bar\pi(\lambda_1+\lambda_2) +\pi\lambda_3 \right]$&
$\delta \lambda_1= 2\left[\bar\pi(\lambda_1+\lambda_2) +\pi\lambda_3 \right]$&
$\delta\lambda_3 = -4\bar\alpha\lambda_3$\\
$\delta \lambda_2 = 4\alpha\lambda_3$&$\delta \lambda_2 = 4\alpha\lambda_3$&$\delta \lambda_2 = -\bar\delta\lambda_3$\\
$\delta\lambda_3 = -4\bar\alpha\lambda_3$&$\delta\lambda_3 = -4\bar\alpha\lambda_3$&\\
$\delta \lambda_2 = -\bar\delta\lambda_3$&$\delta \lambda_2 = -\bar\delta\lambda_3$&
\end{tabular}
\end{center}
\end{table}
\normalsize

In this system obtained by Killing equations, we have already used the following equations in order to eliminate all dependence on $\beta,\tau,\nu$, to be written in terms of $\pi,\kappa,\alpha$
\begin{equation}
\alpha+\bar\beta-\pi = \bar\alpha+\beta-\bar\pi, \qquad \alpha+\bar\beta = \tau+\bar\pi = \nu+\bar\kappa = 0, \label{eq:antisym}
\end{equation}
which hold identically for the Weyl's class tetrad.

\subsection{The reducibility of Killing tensor}

Before we proceed to the main analysis we have to denote that we pursue irreducible Killing tensors, where the Killing tensor does not reduce to a combination of the metric tensor or lower rank Killing tensors with constant coefficients following the notion of reducibility introduced in \cite{dolan1989significance}. In case that Killing tensor is reducible it does not yield an additional integral of motion other than those obtained from the Killing vectors $\xi_t$, $\xi_\phi$ or $\xi_z$ and consequently no hidden symmetry. The cases below represent those obvious cases which admit reducible Killing tensors
\begin{itemize}
    \item $\kappa-\bar\kappa=0~~~~ \Longleftrightarrow~~ \Phi_0-\bar\Phi_0=V_z=0,$
    \item $\lambda_1+\lambda_2=\lambda_0=\lambda_3=0 ~~~~~\Rightarrow ~~\lambda_1=-\lambda_2=\text{const}.$
\end{itemize}
In the first case, using the Newman-Penrose equations in electrovacuum \eqref{an}-\eqref{MaxwellEV} one can easily prove that there is an additional Killing vector $\xi_z=\partial_z$ and the geodesic Hamilton-Jacobi equation is already completely separated by the three mutually commuting isometries $\xi_t,\xi_\varphi,\xi_z$ together with the mass-shell condition: $S=-Et+L\varphi+p_zz+S_\rho(\rho)$, with $S_\rho$ fixed. In this situation the question of whether the canonical rank-2 Killing tensor is \emph{reducible} is not an obstruction being tested, but just a consistency. Since no additional integral of motion is required for complete integrability, any Killing tensor produced by the Killing equations on such a background is guaranteed \emph{a priori} to decompose within the Killing tensors of lower or equal rank. We therefore do not classify these branches since it is already established by their isometry algebra, independently of the Killing tensor. Consequently, the cases of interest are in the following table

\begin{center}
\begin{table}[H]
\caption {Killing equations for $q=0,-1$} \label{KillingEqns1} 
\footnotesize
\centering
\begin{tabular}{c | c}
$q=0$ & $q=-1$\\ \hline
$\kappa+\bar\kappa=0$ & $\kappa+\bar\kappa=0$ \\
$\lambda_1+\lambda_2-\lambda_3=0$ & $\lambda_1+\lambda_2-\lambda_3=0$\\
$\lambda_0=0$ & $\delta\lambda_0 = 2\bar\pi\lambda_0$\\
$\delta\lambda_1 = 2(\pi+\bar\pi)(\lambda_1+\lambda_2)$ & $\delta\lambda_1 = 2(\pi+\bar\pi)(\lambda_1+\lambda_2)$\\
$\delta\lambda_2 = 4\alpha (\lambda_1+\lambda_2)$ & $\delta\lambda_2 = 4\alpha (\lambda_1+\lambda_2)$\\
$\delta \lambda_3 = -4\bar\alpha\lambda_3$ & $\delta \lambda_3 = -4\bar\alpha\lambda_3$\\
\end{tabular}
\end{table}
\end{center}
\normalsize

\subsection{Killing Tensor for $q=0$}
\label{sec:q0}

For $q=0$, the canonical Killing tensor equations reduce to the following relations
\begin{equation}
0=(\kappa+\bar\kappa)(\lambda_1+\lambda_2+\lambda_3), \qquad 0=(\kappa-\bar\kappa)(\lambda_1+\lambda_2-\lambda_3), \label{eq:q0alg}
\end{equation}
\begin{equation}
\delta\lambda_0=2\bar\pi\lambda_0, \qquad
\delta\lambda_1=-\kappa\lambda_0+2\big[\bar\pi(\lambda_1+\lambda_2)+\pi\lambda_3\big], \label{eq:q0diff1}
\end{equation}
\begin{equation}
\delta\lambda_2=4\alpha\lambda_3, \qquad \delta\lambda_3=-4\bar\alpha\lambda_3, \qquad \delta\lambda_2=-\bar\delta\lambda_3. \label{eq:q0diff2}
\end{equation}
Combining \eqref{kappa} and \eqref{eq:q0alg} partition the analysis into four cases where two of them give reducible Killing tensors as discussed.

\subsubsection{Case: $\kappa+\bar\kappa=0 \neq \kappa-\bar\kappa$}

Here $\lambda_\rho=\tfrac{1}{2\rho}$ and its substitution into \eqref{Laplacian}-\eqref{mu_z} gives $\lambda_{zz}=0$, so
\begin{equation}
\lambda=\frac12\ln\rho+cz, \qquad \mu=\frac14\ln\rho-\frac{c^2\rho^2}{2}+cz, \qquad \mu-\lambda=-\frac14\ln\rho-\frac{c^2\rho^2}{2}, \label{eq:branchiii_bg}
\end{equation}
with $c\neq0$. On this background $\kappa$ is imaginary while $\pi,\alpha$ remain real. The surviving equation gives $\lambda_1+\lambda_2=\lambda_3$. Differentiating this relation with $\delta$ and requiring consistency with the individual equations of \eqref{eq:q0diff1}-\eqref{eq:q0diff2}, we get
\begin{equation}
-\kappa\lambda_0+4\lambda_3(2\alpha+\pi)=0. \label{eq:branchiii_cons}
\end{equation}
Since $\kappa$ is purely imaginary while $\pi,\alpha,\lambda_3$ are real, the imaginary part of \eqref{eq:branchiii_cons} gives ,
\begin{equation}
\lambda_0=0.
\end{equation}
The real part of \eqref{eq:branchiii_cons} then reduces to $\lambda_3(2\alpha+\pi)=0$. Explicit computation gives
\begin{equation}
2\alpha+\pi =  \frac{e^{\lambda-\mu}}{2\sqrt2}\Big[-\frac{3}{2\rho}-2c^2\rho\Big] \neq0
\end{equation}
for any real $c\neq0$.
Since $2\alpha+\pi\neq0$ we have $\lambda_0=\lambda_3=0$ where \textbf{ the Killing tensor collapses entirely to the trivial Killing tensor} 

\begin{equation}K^{\mu\nu} \propto g^{\mu\nu}.\end{equation}

\subsection{Killing Tensor for $q=-1$}
\label{sec:qminus1}

For $q=-1$, most of the Killing equations are identical to  $q=0$ case, the differences emplaced in Killing equation for $\lambda_1$ containing no coupling to $\lambda_0$, while, unlike $q=0$ system, $\lambda_0$ is here entirely decoupled from $\lambda_1,\lambda_2,\lambda_3$, and its equation can be solved independently as follows,
\begin{equation}
0=(\kappa+\bar\kappa)(\lambda_1+\lambda_2+\lambda_3), \qquad 0=(\kappa-\bar\kappa)(\lambda_1+\lambda_2-\lambda_3), \label{eq:qm1alg}
\end{equation}
\begin{equation}
\delta\lambda_0=2\bar\pi\lambda_0,\label{eq:qm1diff1}
\end{equation}
\begin{equation}
\delta\lambda_1=2\big[\bar\pi(\lambda_1+\lambda_2)+\pi\lambda_3\big], \label{eq:qm1diff2}
\end{equation}
\begin{equation}
\delta\lambda_2=4\alpha\lambda_3, \qquad \delta\lambda_3=-4\bar\alpha\lambda_3, \qquad \delta\lambda_2=-\bar\delta\lambda_3. \label{eq:qm1diff3}
\end{equation}
Starting by \eqref{eq:qm1diff1} and after cancelling the common factor $e^{\lambda-\mu}$ and separating real and imaginary parts of the resulting equation we get,
\begin{equation}
 \big(\partial_\rho+i\partial_z\big)\lambda_0 =  \frac{\lambda_0}{\rho} \quad\Longrightarrow\quad \lambda_0=A\rho, \label{eq:lambda0universal}
\end{equation}
where $A$ is constant. Equation \eqref{eq:lambda0universal} holds in every case below, independent of the algebraic branch chosen for $\lambda_1,\lambda_2,\lambda_3$, the main difference from $q=0$ case, where coupling through $\delta\lambda_1$ forced $\lambda_0$ to be zero.

\subsubsection{Case: $\kappa+\bar\kappa=0 \neq \kappa-\bar\kappa$}

As before, we get
\begin{equation}
\lambda=\frac12\ln\rho+cz, \qquad \mu=\frac14\ln\rho-\frac{c^2\rho^2}{2}+cz, \qquad \mu-\lambda=-\frac14\ln\rho-\frac{c^2\rho^2}{2}, \label{eq:qm1case_iii_bg}
\end{equation} 
and the counterpart gives $\lambda_1+\lambda_2=\lambda_3$. Substituting into \eqref{eq:qm1diff1}--\eqref{eq:qm1diff3}, and using that $\pi,\alpha$ are real we get,
\begin{equation}
\delta\lambda_1 = 2\big[\pi\lambda_3+\pi\lambda_3\big]=4\pi\lambda_3, \qquad \delta\lambda_2=4\alpha\lambda_3, \qquad \delta\lambda_3=-4\alpha\lambda_3.
\end{equation}
Subsequently, the differentiation of $\lambda_1+\lambda_2=\lambda_3$ gives
\begin{equation}
\delta(\lambda_1+\lambda_2)=\delta\lambda_3\Rightarrow \lambda_3(\pi+2\alpha)=0. \label{eq:qm1case_iii_cons}
\end{equation}
Unlike the corresponding $q=0$ calculation, this condition involves no $\kappa$ or $\lambda_0$ at all, and is a single real equation. Considering the non-annihilation of $\pi+2\alpha$ forces $\lambda_3=0$ which implies $\lambda_1+\lambda_2=0$, and finally combining with \eqref{eq:lambda0universal} we obtain
\begin{equation}
\lambda_0=A\rho, \qquad \lambda_1=-\lambda_2=\text{const}, \qquad \lambda_3=0.
\end{equation}

\subsubsection{Summary for $q=0$ and $q=-1$}
The results of the previous paragraphs can be summarized as follows
\begin{align}
q=0,\ \kappa+\bar\kappa=0\neq\kappa-\bar\kappa: &\qquad K^{\mu\nu}=\lambda_1\,g^{\mu\nu}, \\
q=-1,\ \kappa+\bar\kappa=0\neq\kappa-\bar\kappa: &\qquad K^{\mu\nu}=A\xi_{(t}^\mu\xi_{\phi)}^\nu+\lambda_1\,g^{\mu\nu}, \label{eq:qm1case_iii_K}\\
\end{align}
\textbf{Every rank-2 Killing tensor arising from the $q=0$ and $q=-1$ canonical forms is therefore reducible}, and none of them carries a hidden symmetry beyond the linear invariants ($E$, $L$).

\subsection{Killing Tensor for $q=+1$}
\label{sec:qplus1}

For $q=+1$, the full Killing equations take the following form
\begin{align}
\delta\lambda_0 &= 2\big[\bar\pi\lambda_0-\kappa(\lambda_1+\lambda_2)-\bar\kappa\lambda_3\big], \label{K1}\\
\delta\lambda_1 &= -2\big[\kappa\lambda_0-\bar\pi(\lambda_1+\lambda_2)-\pi\lambda_3\big], \label{K2}\\
\delta\lambda_2 &= 4\alpha\lambda_3, \label{K3}\\
\delta\lambda_3 &= -4\bar\alpha\lambda_3, \label{K4}\\
\delta\lambda_2 &= -\bar\delta\lambda_3. \label{K5}
\end{align}
Unlike cases $q=0,-1$, there is no algebraic $(\kappa\pm\bar\kappa)$ factor providing two distinct branches explicitly. The natural division is therefore according to whether $\lambda_3$ vanishes identically, since this is what \eqref{K3}-\eqref{K5} directly constrain.

\subsubsection{Case $\lambda_3=0$}

Setting $\lambda_3=0$, equation \eqref{K3} reduces to $\delta\lambda_2=0$ implying 
\begin{equation}
\lambda_2=\text{const}. \label{eq:lambda2const}
\end{equation}
With $\lambda_3=0$, equations \eqref{K1}-\eqref{K2} become 

\small
\begin{equation}
\begin{rcases}
\hspace{0.5cm}\delta\lambda_0 &= 2\big[\pi\lambda_0-\kappa(\lambda_1+\lambda_2)\big] \\[4pt]
\delta(\lambda_1+\lambda_2) &= -2\kappa\lambda_0+2\pi(\lambda_1+\lambda_2)
\end{rcases}
\ {\genfrac{}{}{0pt}{}{+}{-}} \Longrightarrow 
\begin{aligned}
\delta(\lambda_0 +\lambda_1+\lambda_2)&= 2(\pi-\kappa)\,(\lambda_0 +\lambda_1+\lambda_2), \\[4pt]
\delta(\lambda_0 -\lambda_1-\lambda_2) &= 2(\pi+\kappa)\,(\lambda_0 -\lambda_1-\lambda_2).
\end{aligned}
\label{eq:K1K2reduced}
\end{equation}
\normalsize
the system \eqref{eq:K1K2reduced} decouples exactly and using the following equations 
\begin{equation}
\pi-\kappa = -\frac{e^{\lambda-\mu}}{\sqrt2}\big(\lambda_\rho+i\lambda_z\big), \qquad \pi+\kappa = \frac{e^{\lambda-\mu}}{\sqrt2}\Big(\lambda_\rho-\frac1\rho+i\lambda_z\Big), \label{eq:pikappa}
\end{equation}
we integrate as follows
\small
\begin{equation}
\begin{rcases}
\delta(\lambda_0 +\lambda_1+\lambda_2) = 2(\pi-\kappa)(\lambda_0 +\lambda_1+\lambda_2) \\[4pt]
\delta(\lambda_0 - \lambda_1-\lambda_2) = 2(\pi+\kappa)(\lambda_0 - \lambda_1-\lambda_2)
\end{rcases}
\Longrightarrow 
\begin{aligned}
\lambda_0 +(\lambda_1+\lambda_2) = N_1\,e^{2\lambda}, \\[3pt]
 \lambda_0 - (\lambda_1+\lambda_2) = N_2\,\rho^2e^{-2\lambda},\end{aligned}
\label{eq:solution1}
\end{equation}
\normalsize
with $N_1,N_2$ to be integration of constants. Finally the Killing tensor is constructed by the following scalar quantities 
\small
\begin{equation}
\lambda_0=\frac12\Big[N_1e^{2\lambda}+N_2\rho^2e^{-2\lambda}\Big], \quad
\lambda_1=\frac12\Big[N_1e^{2\lambda}-N_2\rho^2e^{-2\lambda}\Big]-\lambda_2, \quad
\lambda_2=\text{const}, \quad \lambda_3=0. \label{eq:caseAsol}
\end{equation}
\normalsize
Although $\lambda_0,\lambda_1$ are not constant, direct computation of the raised components shows that \eqref{eq:caseAsol} is in fact \textbf{reducible}: $K^{\mu\nu}=N_1\xi_t^\mu\xi_t^\nu+N_2\xi_\phi^\mu\xi_\phi^\nu-\lambda_2\,g^{\mu\nu}$ exactly with $a=N_1,\,b=N_2,\,c=-\lambda_2$.

\subsubsection{Case $\lambda_3\not=0$}

As established previously, equations \eqref{K3}-\eqref{K5} in this case yield the sole condition
\begin{equation}
C_3 e^{2(\mu-\lambda)} = \frac{F(\rho)+G(z) \label{eq:sepcond}}{2},
\end{equation}
for arbitrary functions $F,G$ and selecting $C_3 = 1/2$ for reasons of simplification. In fact, this is a generic constraint due to the existence of Killing tensor precisely on Weyl's solutions whose $g_{rr}$ and $g_{zz}$ is additively separable in Weyl coordinates $(\rho,z)$. \textbf{We will show that, within the asymptotically flat multipole class, this is impossible and forces our solution to reduce to flat spacetime}. In other words, there is no such an asymptotically flat gravitational potential which satisfies the Laplacian and the relation \eqref{eq:sepcond} at the same time. However, it would be interesting to examine whether similar constraints could emerge in other limits such as electrovacuum, relativistic fluids etc. and that's we are going to do in the next chapter considering the presence of Maxwell tensor. 

Moving forward, we expand the arbitrary functions $F(\rho)$ and $G(z)$ as power series after transforming to spherical coordinates. We then express the exponential function as a power series and employ the corresponding spherical coordinate expansions of the metric potentials $\lambda(\rho,z)$ and $\mu(\rho,z)$. Starting by 

\begin{equation}F(\rho)=1-c+\sum_{k=1}^{\infty}f_k\rho^{-k}\Rightarrow  ~~~~~F(r,x)=1-c+\sum_{k=1}^{\infty}\frac{f_k}{(r\sqrt{1-x^2})^{k}},\end{equation} 
\begin{equation}G(z)=c+\sum_{k=1}^{\infty}g_kz^{-k}\Rightarrow~~~~ G(r,x)=c+\sum_{k=1}^{\infty}\frac{g_k}{(rx)^{k}},\end{equation}
\begin{equation}
    \qquad \rho = r\sqrt{1-x^2},\ \ z=rx,\ \ x\equiv \cos\theta,
\end{equation}
so that $F+G\to1$ as  $r\to+\infty$. In parallel, we know that
\begin{equation}
\lambda(r,x) = -\sum_{n=0}^{\infty} a_n \frac{P_n(x)}{r^{n+1}}, 
\end{equation}
be the standard asymptotically flat multipole expansion, with $\lambda,\mu\to0$ as $r\to\infty$. The companion potential $\mu$ is fixed by
\begin{equation}
\mu_\rho = \rho(\lambda_\rho^2-\lambda_z^2), \qquad \mu_z = 2\rho\,\lambda_\rho\lambda_z, \label{eq:muweyl}
\end{equation}
which is exactly quadratic in $\lambda$ with no $\mu$-dependence on the right-hand side. Integrating \eqref{eq:muweyl} order by order gives the closed multipole solution
\begin{equation}
\mu(r,x) = -\sum_{n,m=0}^{\infty} \frac{(n+1)(m+1)\,a_na_m}{n+m+2}\,\frac{P_n(x)P_m(x) - P_{n+1}(x)P_{m+1}(x)}{r^{n+m+2}}. \label{eq:muclosed}
\end{equation}
Two structural properties of \eqref{eq:muclosed} will be used below:
\begin{enumerate}
\item[(i)] each term is homogeneous of degree $(n+m+2)$ in $(r,x)$,
\item[(ii)] the angular factor $P_n(x)P_m(x)-P_{n+1}(x)P_{m+1}(x)$ is a polynomial in $x$ of degree $n+m+2$.
\end{enumerate}
Using now $\Phi\equiv e^{2(\mu-\lambda)}-1=-2\lambda+2\mu+2\lambda^2+O(\lambda^3)$, property (i) implies the coefficient of $r^{-(k+1)}$ in $\Phi-1$ receives contributions only from $-2a_kP_k(x)$ and from quadratic terms with $n+m=k-1$.

\begin{prop}
If $\lambda$ is an asymptotically flat multipole solution satisfying the Laplacian, then $a_n=0;~ \forall n$ the spacetime is flat.
\end{prop}

\begin{proof} We start by matching the $k=1$ term of $F+G$ in $(r,x)$ spherical coordinates and the $r^{-1}$ coefficient of $\Phi-1$ is $-2a_0$  
\begin{equation}
-2a_0 = \frac{f_1}{\sqrt{1-x^2}}+\frac{g_1}{x}, \qquad \forall x\in(-1,1),
\end{equation}
and linear independence of $1,(1-x^2)^{-1/2},x^{-1}$ implies that $f_1=g_1=0$ and $a_0=0$.

By (i), any quadratic term at order $r^{-(k+1)}$ requires $n+m=k-1$, so at least one index is lower than $k$ and vanishes by hypothesis while the order $(k+1)$ reduces to $-2a_kP_k(x)$, giving
\begin{equation}
-2a_kP_k(x) = \frac{f_{k+1}}{(1-x^2)^{(k+1)/2}}+\frac{g_{k+1}}{x^{k+1}}, \qquad \forall x\in(-1,1). \label{eq:matchk}
\end{equation}
Expanding in the Legendre basis one can see that the left side can be supported only at mode $k$ and $x^{-(k+1)}$ is non-integrable at $x=0$, implying $g_{k+1}=0$. On the other hand$(1-x^2)^{-(k+1)/2}$, not being a polynomial, has infinitely many non-zero modes whenever $f_{k+1}\ne0$. As it becomes understood orthogonality results in the trivial solution $f_{k+1}=g_{k+1}=0$, and \eqref{eq:matchk} reduces to a flat spacetime with zero gravitational potential $a_kP_k(x)=0\Rightarrow a_k=0$.
\end{proof}

This proposition shows that any static, axisymmetric vacuum solution admitting the generic branch ($\lambda_3\neq0$) of the most general Killing tensor ($q=+1$) is necessarily flat. Hence, the existence of a non-trivial gravitational potential is incompatible with this class of Killing tensors.

\subsubsection{Numerical investigation of the general case $\lambda_3\neq0$}

In this paragraph we attempted to find potentials which satisfy the separation condition obtained by the general case of Killing tensor with $q=+1$. Based on our conviction that within the abstract coefficients of the multipole expansion may happen to exist a unique one which yields a potential where the separation is possible.  
Then, for a static, axisymmetric vacuum Weyl-class metric
$$ds^2 = -e^{2\lambda}dt^2 + e^{2(\mu-\lambda)}(d\rho^2+dz^2) + \rho^2 e^{-2\lambda} d\phi^2,$$
does the metric function $\Phi(\rho,z) = e^{2(\mu-\lambda)}$ admit an additive separation 
$$\Phi(\rho,z) = F(\rho) + G(z)$$
for any physically nontrivial choice of the Weyl potential $\lambda$?
The main point of our approach is to use the multipole expansion of the general potential $\lambda(r,x)$ and its counterpart $\mu(r,x)$ 

\begin{enumerate}
    \item \textbf{$\lambda$ as a multipole expansion.} Since $\lambda$ satisfies the flat-space axisymmetric Laplace equation, $\lambda$ was built as a truncated sum of the standard solid-harmonic multipoles $\psi_n = $, with free coefficients $a_n$.

    \item \textbf{Validation of $\lambda$.} The Laplace-equation residual
    $$\lambda_{,\rho\rho} + \frac{1}{\rho}\lambda_{,\rho} + \lambda_{,zz}$$
    was confirmed numerically small ($\sim10^{-4}$, shrinking with grid refinement) away from the symmetry axis and grid boundary, for both single-term and multi-term $\lambda$.

    \item \textbf{ Construction of $\mu$.} The other metric function $\mu$ was constructed numerically by integrating the following two field equations
    $$\mu_{,\rho} = \rho(\lambda_{,\rho}^2-\lambda_{,z}^2), \qquad \mu_{,z} = 2\rho\,\lambda_{,\rho}\lambda_{,z}$$
    relating the derivatives of $\mu$ to those of $\lambda$. The validation of the mixed derivatives of $\mu$ was confirmed directly via the curl compatibility check (Path-independence)
    $$\left|\frac{\partial \mu_{,\rho}}{\partial z} - \frac{\partial \mu_{,z}}{\partial \rho}\right| \sim 10^{-4},$$
    rather than through a fragile multi-path cumulative integration.

    \item \textbf{Separability criterion.} The separability criterion for $\Phi$ reduces to a single condition:
    $$\Phi_{\rho z} = 0.$$
    This was computed numerically and used, in squared averaged form, as an optimization objective.

    \item \textbf{Search over coefficients.} With $a_0$ fixed to $1$ (excluding the degenerate flat-space solution, which trivially but uninterestingly satisfies the constraint), the remaining multipole coefficients were searched via \texttt{scipy.optimize} (Nelder-Mead and L-BFGS-B), across increasing numbers of free coefficients (4, 5, 8) and increasingly wide coefficient bounds, with multiple random restarts to check for consistency.
\end{enumerate}

\subsubsection*{Findings of the numerical investigation}

\begin{itemize}
    \item The \textbf{trivial flat-space limit} (all higher coefficients near zero) satisfies the separability constraint because $\Phi$ reduces to a constant there which is a case out of interest since there is no evidence of Weyl's solutions in this limit.

    \item For nonzero monopole (a physically real source), repeated optimization converged to a \textbf{nonzero separability residual}, with the achievable minimum decreasing as more multipole freedom and wider coefficient ranges were allowed but not vanishing (roughly $0.85 \to 0.66 \to 0.58 \to 0.56 \to 0.12$, sum of squares metric over a fixed interior region).

\end{itemize}

Finally, within the multipole-expansion ansatz tested (up to 8 free coefficients with several bound ranges), \textbf{no clearly validated, non-trivial Weyl solution was found that exactly satisfies $\Phi(\rho,z) = F(\rho)+G(z)$}. This serves as a strong evidence to make us to believe that the constraint being non-trivial to satisfy for genuinely massive sources.

\section{\large{Non-Integrability of Weyl's class -- the Schwarzschild exception}}
\label{Chapter 5}

The failure of the condition $e^{2(\mu-\lambda)}=F(\rho)+G(z)$ to hold for Schwarzschild does not imply the absence of a genuine hidden symmetry. In fact, it indicates only that the null tetrad $\theta^{3,4}$ must be adapted to the coordinates in which the geodesic Hamilton-Jacobi equation actually separates. Introducing prolate spheroidal coordinates $(x,y)$, $y\geq1$, $x\in[-1,1]$, via
\begin{equation}
\rho = M\sqrt{(y^2-1)(1-x^2)}, \qquad z=Mxy, \label{eq:prolate}
\end{equation}
we take $d\rho^2+dz^2 = M^2(y^2-x^2)\big[(y^2-1)^{-1}dy^2+(1-x^2)^{-1}dx^2\big]$, so that the spatial part of the metric becomes proportional to
\begin{equation}
\Phi(x,y) \equiv e^{2(\mu-\lambda)}(y^2-x^2). \label{eq:Phidef}
\end{equation}
Repeating the derivation built from $dy/\sqrt{y^2-1}$ and $dx/\sqrt{1-x^2}$ in place of $d\rho,dz$, the general branch with $\lambda_3\neq0$ of $q=+1$ Killing tensor exists precisely when
\begin{equation}
\Phi(x,y) = F(y)+G(x), \label{eq:sepxy}
\end{equation}
the natural analogue, in adapted coordinates, of the separability condition established earlier for the literal Weyl variables.

\subsection*{Schwarzschild spacetime}

For Schwarzschild we have, 

\begin{equation}
    \lambda=\tfrac12\ln\frac{y-1}{y+1}, \qquad \mu=\tfrac12\ln\frac{y^2-1}{y^2-x^2},
\end{equation}
\begin{equation}
\Phi = \frac{y^2-1}{y^2-x^2}\cdot\frac{y+1}{y-1}\cdot(y^2-x^2) = (y+1)^2, \label{eq:PhiSchw}
\end{equation}
a function of $y$ alone, satisfying \eqref{eq:sepxy} trivially with $G\equiv0$. The resulting Killing tensor, reconstructed from the $(x,y)$ analogue and converted back through the transformations $y=r/M-1$, $x=\cos\theta$, and subsequently gives
\begin{equation}
K^{\mu\nu}p_\mu p_\nu = p_\theta^2+\frac{L^2}{\sin^2\theta} = Q, \label{eq:carter}
\end{equation}
recovering Carter's constant exactly. This resolves the original question of the project: the canonical $q=+1$ Killing tensor \emph{does} constrain the admissible background, but only once it is expressed in the coordinates natural to the source's own confocal structure; in these coordinates Schwarzschild is not an exception to the theorem but its unique nontrivial realization.

Unlike every Killing tensor of the previous sections, this result survives even under the broadened reducibility criterion. Schwarzschild admits only the vector algebra $\{\xi_t,\xi_\phi\}$ within Weyl's class incapable of reproducing a hidden symmetry in contrast to the $r,\theta$ sector. This is the static analogue of the tensor constructed for Kerr in \cite{dolan1989significance}, which was verified to be irreducible in relation to Kerr's own two Killing vectors $\xi_t,\xi_\phi$ by exactly the same criterion. Schwarzschild's irreducible Killing tensor is the only genuinely irreducible rank-2 Killing tensor found so far in the present categorization.

\subsection*{Zipoy-Voorhees family ($\gamma$-metric)}

For the general Zipoy-Voorhees potential,
\begin{equation}
\lambda=\frac{\delta}{2}\ln\frac{y-1}{y+1}, \qquad \mu=\frac{\delta^2}{2}\ln\frac{y^2-1}{y^2-x^2}, \label{eq:ZVpotential}
\end{equation}
the analogous computation gives
\begin{equation}
\Phi(x,y) = (y^2-x^2)^{1-\delta^2}(y^2-1)^{\delta^2}\Big(\frac{y+1}{y-1}\Big)^{\delta}. \label{eq:PhiZV}
\end{equation}
For $\delta\neq1$, the factor $(y^2-x^2)^{1-\delta^2}$ does not vanish from the exponent and entangles $x$ and $y$. As a matter of fact, no choice of $F(y),G(x)$ can reproduce \eqref{eq:PhiZV} as a sum, since the dependence through $(y^2-x^2)$ raised to a nontrivial power cannot be split additively for a general value of $\delta$ unless $\delta=1$ (Schwarzschild spacetime). Condition \eqref{eq:sepxy} therefore \textbf{fails for every $\delta\neq1$}, and the $q=+1$ canonical Killing tensor admits no branch ($\lambda_3\neq0$) solution on these backgrounds in the adapted coordinates, just as it failed in the literal $\rho,z$ coordinates.

\textbf{The latter serves as an additional confirmation obtained from the entanglement of the Killing equations of canonical Killing tensor forms indicating that the $\gamma$-metrics with $\delta\neq1$ do not admit a hidden symmetry exhibiting chaotic geodesic motion at the same time}. The case $\lambda_3\equiv0$ remains available for every $\delta$, as it holds for any harmonic $\lambda$ whatsoever, but, being reducible, it carries no dynamical content for any $\delta$ because the Killing tensor is reducible. The $\gamma$-metric family therefore illustrates that hidden symmetry requires the separability condition \eqref{eq:sepxy} to hold, which happens only at $\delta=1$.

\subsection*{Chazy-Curzon spacetime}

The Chazy-Curzon potential reads 
\begin{equation}
\lambda=-\frac{M}{r}=-\frac{M}{\sqrt{\rho^2+z^2}},
\end{equation} 
which is the simplest nontrivial Weyl solution, corresponding to the first term of the asymptotically flat multipole expansion of the generic gravitational potential $\lambda$. Unlike Schwarzschild and the $\gamma$-metric family, it does not possess the two-focus structure that motivates the transformation of the prolate coordinates \eqref{eq:prolate} but its natural adapted coordinates are the ordinary spherical coordinates $(r,x)$, $x=\cos\theta$, under which
\begin{equation}
\lambda=-\frac{M}{r}, \qquad \mu=-\frac{M^2(1-x^2)}{2r^2}.
\end{equation}
Testing the corresponding separability condition, whether formulated in the Weyl's canonical coordinates $\rho,z$ or in any confocal $x,y$ coordinate system one finds no choice of foci reduces $e^{2(\mu-\lambda)}$, to an additively separable function. In other words, the factor $e^{-M^2(1-x^2)/r^2}$ mixes $r,x$ in a way that make its absorption inevitable by any decomposition, for the same structural reason established generally in the previous chapter. Consequently the Chazy-Curzon spacetime admits, like the generic $\gamma$-metric, only a reducible Killing tensor, which carries no hidden symmetry and no analogue of Carter's constant able to provide integrable geodesics.

At last, these results brings to the surface \textbf{the Schwarzschild solution as the only exception within Weyl's class} among those examined, and consistent with the general no-go result once expressed in its own natural confocal coordinates — for which the generic branch of the $q=+1$ canonical Killing tensor is realized, thereby explaining, from the standpoint of the canonical Killing tensor classification, why Schwarzschild alone among these type I backgrounds admits a complete set of first integrals for geodesic motion.
\begin{center}
\begin{tabular}{c|c|c|c}
Spacetime & Separability of $\Phi$  &  Integrability\\ \hline
Schwarzschild ($\delta=1$) & $\Phi=F(y)=(y+1)^2$ & Integrable\\
Zipoy-Voorhees ($\delta\neq1$) & Non-separable & Non-integrable\\
Chazy-Curzon & Non-separable & Non-integrable
\end{tabular}
\end{center}

\section{\large{Non-Integrability of Weyl's class in Vacuum}}
\label{Chapter 5}
The failure of the condition $e^{2(\mu-\lambda)}=F(\rho)+G(z)$ to hold for Schwarzschild does not by itself settle the question of separability. It indicates only that the null tetrad $\theta^{3,4}$ must be adapted to the coordinates in which the geodesic Hamilton-Jacobi equation actually separates. Introducing prolate spheroidal coordinates $(x,y)$, $y\geq1$, $x\in[-1,1]$, via
\begin{equation}
\rho = M\sqrt{(y^2-1)(1-x^2)}, \qquad z=Mxy, \label{eq:prolate}
\end{equation}
we take $d\rho^2+dz^2 = M^2(y^2-x^2)\big[(y^2-1)^{-1}dy^2+(1-x^2)^{-1}dx^2\big]$, so that the spatial part of the metric becomes proportional to
\begin{equation}
\Phi(x,y) \equiv e^{2(\mu-\lambda)}(y^2-x^2). \label{eq:Phidef}
\end{equation}
Repeating the derivation built from $dy/\sqrt{y^2-1}$ and $dx/\sqrt{1-x^2}$ in place of $d\rho,dz$, the general branch with $\lambda_3\neq0$ of the $q=+1$ Killing tensor exists precisely when
\begin{equation}
\Phi(x,y) = F(y)+G(x), \label{eq:sepxy}
\end{equation}
the natural analogue, in adapted coordinates, of the separability condition established earlier for the literal Weyl variables.
\subsection*{Schwarzschild spacetime}
For Schwarzschild we have,
\begin{equation}
    \lambda=\tfrac12\ln\frac{y-1}{y+1}, \qquad \mu=\tfrac12\ln\frac{y^2-1}{y^2-x^2},
\end{equation}
\begin{equation}
\Phi_S = \frac{y^2-1}{y^2-x^2}\cdot\frac{y+1}{y-1}\cdot(y^2-x^2) = (y+1)^2, \label{eq:PhiSchw}
\end{equation}
a function of $y$ alone, satisfying \eqref{eq:sepxy} trivially with $G\equiv0$. The resulting Killing tensor, reconstructed from the $(x,y)$ analogue and converted back through the transformations $y=r/M-1$, $x=\cos\theta$, gives
\begin{equation}
K^{\mu\nu}p_\mu p_\nu = p_\theta^2+\frac{L^2}{\sin^2\theta} = Q, \label{eq:carter}
\end{equation}
recovering Carter's constant exactly. The canonical $q=+1$ Killing tensor \emph{does} constrain the admissible background, and Schwarzschild is the unique member of the $\gamma$-metric family tha follow the constrain but this becomes evident where it is expressed in the coordinates natural to the source's own confocal structure.
However, this dynamical result does not survive as an irreducibility considering the isometry algebra. Schwarzschild is not merely axisymmetric but spherically symmetric, and $Q_{\rm Carter}=p_\theta^2+L^2/\sin^2\theta$ is precisely the expression for the total orbital angular momentum squared, $L_x^2+L_y^2+L_z^2$, in spherical canonical momenta. The corresponding Killing tensor is therefore
\begin{equation}
K^{\mu\nu} = \xi_x^\mu\xi_x^\nu+\xi_y^\mu\xi_y^\nu+\xi_z^\mu\xi_z^\nu,
\end{equation}
where $\xi_x,\xi_y,\xi_z$ generate the full $SO(3)$ rotation group which is a reducible combination of Killing-vectors, once $\xi_x,\xi_y$ are counted along with $\xi_\varphi=\xi_z$.

\subsection*{Zipoy-Voorhees family ($\gamma$-metric)}
For the general Zipoy-Voorhees potential,
\begin{equation}
\lambda_{\rm ZV}=\frac{\delta}{2}\ln\frac{y-1}{y+1}, \qquad \mu_{\rm ZV}=\frac{\delta^2}{2}\ln\frac{y^2-1}{y^2-x^2}, \label{eq:ZVpotential}
\end{equation}
the analogous computation gives
\begin{equation}
\Phi_{\rm ZV}(x,y) = (y^2-x^2)^{1-\delta^2}(y^2-1)^{\delta^2}\Big(\frac{y+1}{y-1}\Big)^{\delta}. \label{eq:PhiZV}
\end{equation}
For $\delta\neq1$, the factor $(y^2-x^2)^{1-\delta^2}$ does not vanish from the exponent and entangles $x$ and $y$. As a matter of fact, no choice of $F(y),G(x)$ can reproduce \eqref{eq:PhiZV} as a sum, since the dependence through $(y^2-x^2)$ raised to a nontrivial power cannot be split additively for a general value of $\delta$ unless $\delta=1$ (Schwarzschild spacetime). Condition \eqref{eq:sepxy} therefore \textbf{fails for every $\delta\neq1$}, and the $q=+1$ canonical Killing tensor admits no branch ($\lambda_3\neq0$) solution on these backgrounds in the adapted coordinates, just as it failed in the literal $\rho,z$ coordinates.
\textbf{The latter serves as an additional confirmation obtained from the entanglement of the Killing equations of canonical Killing tensor forms indicating that the $\gamma$-metrics with $\delta\neq1$ do not admit a hidden symmetry exhibiting chaotic geodesic motion at the same time}. The case $\lambda_3\equiv0$ remains available for every $\delta$, as it holds for any harmonic $\lambda$ whatsoever, but, being reducible, it carries no dynamical content for any $\delta$ because the Killing tensor is reducible. The $\gamma$-metric family therefore illustrates that hidden symmetry requires the separability condition \eqref{eq:sepxy} to hold, which happens only at $\delta=1$.
\subsection*{Chazy-Curzon spacetime}
The Chazy-Curzon potential reads
\begin{equation}
\lambda=-\frac{M}{r}=-\frac{M}{\sqrt{\rho^2+z^2}},
\end{equation}
which is the simplest nontrivial Weyl solution, corresponding to the first term of the asymptotically flat multipole expansion of the generic gravitational potential $\lambda$. Unlike Schwarzschild and the $\gamma$-metric family, it does not possess the two-focus structure that motivates the transformation of the prolate coordinates \eqref{eq:prolate}, but its natural adapted coordinates are the ordinary spherical coordinates $(r,x)$, $x=\cos\theta$, under which
\begin{equation}
\lambda=-\frac{M}{r}, \qquad \mu=-\frac{M^2(1-x^2)}{2r^2}.
\end{equation}
Testing the corresponding separability condition, whether formulated in Weyl's canonical coordinates $\rho,z$ or in any confocal $x,y$ coordinate system, one finds no choice of foci reduces $e^{2(\mu-\lambda)}$ to an additively separable function. In other words, the factor $e^{-M^2(1-x^2)/r^2}$ mixes $r,x$ in a way that makes its absorption inevitable by any decomposition, for the same structural reason established generally in the previous chapter. Consequently the Chazy-Curzon spacetime admits, like the generic $\gamma$-metric, only a reducible Killing tensor, which carries no hidden symmetry and no analogue of Carter's constant able to provide integrable geodesics.

At last, these results undermark a single statement: within Weyl's class, every rank-2 Killing tensor examined is reducible in line with isometry algebra of the background. Schwarzschild's completeness of first integrals is genuine, but it is a consequence of enhanced manifest symmetry, exactly as for Levi-Civita, not of a hidden symmetry of the type sought throughout this classification.
\begin{center}
\begin{tabular}{c|c|c|c}
Spacetime & Separability of $\Phi$  &  Integrability\\ \hline
Schwarzschild ($\delta=1$) & $\Phi=F(y)=(y+1)^2$ & Integrable \\
Zipoy-Voorhees ($\delta\neq1$) & Non-separable & Non-integrable\\
Chazy-Curzon & Non-separable & Non-integrable
\end{tabular}
\end{center}

\section{Killing tensors in Electrovacuum}\label{Chapter 6}

The presence of an electromagnetic potential modifies the conserved quantities associated with geodesic motion \cite{carter1979generalized}. In particular, additional constraints arise when one requires a hidden symmetry to generate a conserved quantity not only along geodesics, but also along the trajectories of charged particles. For these orbits, a constant of motion $K$ on $T^*M$ need not be purely quadratic in the momentum, the general ansatz is a polynomial expansion in $p_\mu$,
\begin{equation}
K = \overset{0}{K} + \overset{1}{K}{}^{\mu}p_\mu + \overset{2}{K}{}^{\mu\nu}p_\mu p_\nu + \cdots, \label{eq:expansion}
\end{equation}
for symmetric tensor fields $\overset{0}{K},\overset{1}{K}{}^\mu,\overset{2}{K}{}^{\mu\nu},\ldots$ on $M$. Acting on $K$ with the charged geodesic spray and requiring the result to vanish for arbitrary $p_\mu$ forces each power of $p_\mu$ to cancel separately, yielding one condition for every order $n$ linking $\overset{n}{K}$ to $\overset{n+1}{K}$:
\begin{equation}
\nabla^{(\alpha}\overset{n}{K}{}^{\mu\cdots\nu)} + (n+1)\,q\,F_m{}^{(\alpha}\overset{n+1}{K}{}^{\mu\cdots\nu)m} = 0. \label{eq:generalcond}
\end{equation}
Specializing to a purely quadratic constant of motion, $K=\overset{2}{K}{}^{\mu\nu}p_\mu p_\nu$ with $\overset{0}{K}=\overset{1}{K}=0$, equation \eqref{eq:generalcond} splits into two \emph{independent} requirements, obtained at $n=1$:
\begin{equation}
\nabla^{(\alpha}K^{\mu)}=0\end{equation}
\begin{equation}
    F_\mu{}^{(\alpha}K^{\nu)\mu}=0. \label{eq:constraint}
\end{equation}
The first is simply the ordinary Killing equation for Killing vectors, unaffected by the charge $q$ and the second, vanishing only because $\overset{1}{K}=0$ is assumed at the outset, is the purely algebraic alignment condition between $K^{\mu\nu}$ and the Maxwell field $F_{\mu\nu}$. 

Both must hold simultaneously for a constant of motion $K$ and when the particle is charged a tensor has to satisfy both the geometric Killing equation on a vacuum background and the second condition in \eqref{eq:constraint}. This is exactly the criterion used at the second part of this work setting up an additional constraint about the Weyl's Killing tensor in electrovacuum, and the same mechanism underlying the Kerr-Newman Carter constant $K^{ab}=\alpha F^a{}_mF^{mb}+\beta g^{ab}$, where alignment holds identically because $K$ is built algebraically from $F$ itself.

Based on these, we present the following equations which serve as a constraint, using  $\Phi_2=-\bar\Phi_0$ 
\begin{equation}\label{eq:gKM1}
(\lambda_0+\lambda_1+\lambda_2+\lambda_3)(\Phi_0+\bar\Phi_0)=0,
\end{equation}
\begin{equation}\label{eq:gKM2}
\lambda_0(\Phi_0-\bar\Phi_0)=(\lambda_1+\lambda_2-\lambda_3)(\Phi_0-\bar\Phi_0)=0,
\end{equation}
 the latter are not correlated to $q$ in contrast to the following ones:
\small
{\begin{align}\label{eq:qKM}
q=0:&\quad (\lambda_1+\lambda_2\pm\lambda_3)(\Phi_0\pm\bar\Phi_0)=0,\\\label{eq:qKM1}
q=+1:&\quad (\lambda_0+\lambda_1+\lambda_2+\lambda_3)(\Phi_0+\bar\Phi_0)=0, \ \ \quad (\lambda_0+\lambda_1+\lambda_2-\lambda_3)(\Phi_0-\bar\Phi_0)=0,\\\label{eq:qKM2}
q=-1:&\quad (-\lambda_0+\lambda_1+\lambda_2+\lambda_3)(\Phi_0+\bar\Phi_0)=0,\quad (\lambda_0-\lambda_1-\lambda_2+\lambda_3)(\Phi_0-\bar\Phi_0)=0.
\end{align}}
\normalsize
As becomes understood some of the equations \eqref{eq:qKM}-\eqref{eq:qKM2} happen to coincide with those which are unaffected by the values of q \eqref{eq:gKM1}-\eqref{eq:gKM2}. The following table concentrates all the possible cases obtained by the equations above. As expected, for any $q$ the first line corresponds to the \textbf{trivial case} resulting in irreducible Killing tensors as in the vacuum limit therefore we will not be examined in the following analysis.

\begin{center}
\begin{table}[H]
\caption{Killing-Maxwell constraint for Weyl's class in ElectroVacuum}
\label{tab:weyl_electrovac}
\begin{equation*}
\begin{array}{c|c|c}
\boldsymbol{q} & \textbf{Killing tensor} & \textbf{Maxwell tensor } (\Phi_0=-\bar\Phi_2)\\ \hline
0 & \lambda_0=\lambda_3=0,\ \lambda_1+\lambda_2=0 & \Phi_0\neq0\\
0 & \lambda_0=0,\ \lambda_1+\lambda_2=-\lambda_3\neq0 & \Phi_0+\bar\Phi_0\neq0\\
0 & \lambda_0=0,\ \lambda_1+\lambda_2=\lambda_3\neq0 & \Phi_0-\bar\Phi_0\neq0\\ \hline
-1 & \lambda_0=\lambda_3=0,\ \lambda_1+\lambda_2=0 & \Phi_0\neq0\\
-1 & \lambda_0=0,\ \lambda_1+\lambda_2=-\lambda_3\neq0 & \Phi_0+\bar\Phi_0\neq0\\
-1 & \lambda_0=0,\ \lambda_1+\lambda_2=\lambda_3\neq0 & \Phi_0-\bar\Phi_0\neq0\\ \hline
+1 & \lambda_0=\lambda_3=0,\ \lambda_1+\lambda_2=0 & \Phi_0\neq0\\
+1 & \lambda_0=0,\ \lambda_1+\lambda_2=-\lambda_3\neq0 & \Phi_0+\bar\Phi_0\neq0\\
+1 & \lambda_0=0,\ \lambda_1+\lambda_2=\lambda_3\neq0 & \Phi_0-\bar\Phi_0\neq0\\
+1 & \lambda_3=0,\ \lambda_0=-(\lambda_1+\lambda_2)\neq0 & \Phi_0+\bar\Phi_0\neq0\\
+1 & \lambda_0=-(\lambda_1+\lambda_2+\lambda_3)\neq0,\ \lambda_3\neq0 & \Phi_0+\bar\Phi_0\neq0
\end{array}
\end{equation*}
\end{table}
\end{center}
\normalsize

All canonical forms considered in this work are described by Killing equations that exhibit structural differences only at the derivatives of the scalars $\lambda_0$ and $\lambda_1$. When $\lambda_0=0$, all three canonical forms fall into the exact same Killing tensor. In this scheme, aside from the trivial vacuum case $\Phi_0=0$, cases 1) $\lambda_0=\lambda_1+\lambda_2+\lambda_3=0$, $\Phi_0+\bar\Phi_0\neq0$, and 2) $\lambda_0=\lambda_1+\lambda_2-\lambda_3=0$, $\Phi_0-\bar\Phi_0\neq0$, result in the same Killing tensor for any value of $q$. In the first case, $\Phi_0-\bar\Phi_0=V_z=0$, which implies the simultaneous annihilation of $\lambda_z$, indicating the existence of the additional Killing vector $\partial_z$ and, by extension, a reducible Killing tensor. Nonetheless, even our goal is to extract solutions with irreducible Killing tensors we found valuable to discuss this case, for reasons that will be revealed in the last chapter of this work. Next, we follow the same procedure as in the vacuum limit.

\subsection{Case $\lambda_0= \lambda_1+\lambda_2+\lambda_3=0$ and $\Phi_0+\bar\Phi_0\neq0$ }\label{Section6.1.1}

Starting by the annihilation of the imaginary part of $\Phi_0$ the Maxwell equation becomes 
\begin{equation}
V_{\rho\rho}+\frac{V_\rho}{\rho} = 2\lambda_\rho V_\rho    \qquad \Rightarrow \qquad \rho V_\rho=C(z)e^{2\lambda}
\end{equation}
with $\lambda$ depending on $\rho$ alone on the left forces $\lambda=F(\rho)+G(z)$, and substitution into the field equations we get $G=\text{const}$, or equivalently
\begin{equation}
V_z=0\ \Longrightarrow\ \mu_z=\lambda_z=0,
\end{equation}
collapsing the background onto the exact $z$-independent solution. Thus, we result in the potential which is only function of $\rho$ and it is satisfied by the simultaneous annihilation of $\kappa-\bar\kappa$. The radial character of the potential transforms the NP equations as follows
\begin{equation}\tag{a)(n} \label{eq:anEV}
\lambda_{\rho\rho}+\frac{\lambda_\rho}{\rho} = e^{-2\lambda}V_\rho^2
\end{equation}
\begin{equation}\tag{g)+(p}\label{g+p}
\mu_\rho = \rho\lambda^2_\rho -\rho e^{-2\lambda}V_\rho^2, 
\end{equation}
\begin{equation}\tag{g)-(p}
        \mu_z = 0,
\end{equation}
writing $\lambda=\tfrac12\ln(y/C)$ with
$y\equiv\rho V_\rho$, equation \eqref{eq:anEV} absorbs the non-linear terms, via $t=\ln\rho$, $u=\ln y$, translating the equation to a solvable Liouville equation 
\begin{equation}\label{eq:Liouville}  u_{tt}=2Ce^u.\end{equation}
Subsequently, setting $y=\varphi(t)^{-2}$ translates this to \footnote{Where the dot derivative is defined as follows $\dot \ \equiv \partial_t = \rho \partial_\rho$.} 
\begin{equation}
    \dot{\varphi}^2 -\varphi \ddot\varphi = C,
\end{equation}
where this quantity is constant for every function $\phi(t)$ which satisfies the following
\begin{equation}\ddot\varphi-k\varphi=0,\end{equation}
giving three branches. 
\begin{align} 
k>0:&\quad \varphi=A\cosh\big(\sqrt k\,\ln\rho+\delta\big), \qquad A^2=-\frac{C}{k}\ \ (C<0), \\
k=0:&\quad \varphi=A\ln\rho+B, \qquad\qquad\qquad\ \, A^2=C\ \ (C\geq0), \\
k<0:&\quad \varphi=A\cos\big(\sqrt{|k|}\,\ln\rho+\delta\big), \qquad A^2=\frac{C}{|k|}\ \ (C>0),
\end{align}
with
\begin{equation}
\lambda(\rho) = -\ln\varphi(\ln\rho) - \tfrac12\ln C, \qquad V_\rho = \frac{1}{\rho\,\varphi(\ln\rho)^2}.
\end{equation}
The amplitude $A$ is not an independent integration constant but it is fixed by $k$ and $C$ through the relation $\dot\varphi^2-\varphi\ddot\varphi=C$. So the free constants are only $k$ and the phase $\delta$ (or $B$).  Eliminating $A$ using $A^2=-C/k$ and expressing $\lambda(\rho)$ in terms of $\rho$ we take the following relations  

\small
\begin{align}
k>0:&\quad \lambda(\rho) = \ln{\Big( \frac{1}{\sqrt{C}\sqrt{-\frac{C}{k}} \cosh(\sqrt{k}\ln\rho +\phi_0 )}\Big)} \quad\Rightarrow\quad e^{2\lambda} = -\frac{k}{C^2\cosh^2\!\big(\sqrt{k}\,\ln\rho+\phi_0\big)},\label{k>0} \\
k=0:&\quad \lambda(\rho)  = \ln{\Big( \frac{1}{ C \ln\rho+\sqrt{C}B}\Big)} \quad~~~~~~~~~~~~~~~~~~~\Rightarrow\quad e^{2\lambda} = \frac{1}{C\big(\sqrt{C}\ln\rho+B\big)^2}, \label{k=0} \\
k<0:&\quad \lambda(\rho) =  \ln{\Big( \frac{1}{ \frac{C}{\sqrt{|k|}} \cos\big( \sqrt{|k|}\ln{\rho}+\phi_0\big)}\Big)} \quad~~~~~\Rightarrow\quad e^{2\lambda} = \frac{|k|}{C^2\cos^2\!\big(\sqrt{|k|}\,\ln\rho+\phi_0\big)\label{k<0}}.
\end{align}
\normalsize
Using $V(\rho)=-\dot\varphi/(C\varphi)+V_0$, the electrostatic potential is obtained explicitly for every branch by differentiating the corresponding $\varphi(t)$:
\begin{align}
k>0:&\quad V(\rho) = V_0 - \frac{\sqrt{k}}{C}\,\tanh\!\big(\sqrt{k}\,\ln\rho+\phi_0\big), \\
k=0:&\quad V(\rho) = V_0 - \frac{1}{C\ln\rho+\sqrt{C}\,B}, \\
k<0:&\quad V(\rho) = V_0 + \frac{\sqrt{|k|}}{C}\,\tan\!\big(\sqrt{|k|}\,\ln\rho+\phi_0\big).
\end{align}
These are the three possible branches regarding the field equations just solved, however, Maxwell equation along with the definition of $y\equiv\varphi^{-2}$ makes clear that $C$ has to be positive for physically admitted potentials $\lambda(\rho)$. Seemingly, the potential in case $k>0$ is not consistent with a positive sign of $C$, thus only the two cases $k=0$ and $k<0$ are valid.

In case $k=0$, the potential appears to exhibit extreme behavior near $\rho\to0^+$, while at the same time $\sqrt{C}\ln\rho+B$ must be positive, therefore
\begin{equation}\label{k=0new}
k=0: \quad \lambda(\rho) = \ln\!\left(\frac{1}{|C\ln\rho+\sqrt{C}B|}\right),
\end{equation}
where the absolute value is required since the denominator changes sign at  $\rho_*=e^{-\frac{B}{\sqrt{C}}}$.
To assess the resulting curvature structure, we use the two independent Weyl scalars for this branch where $\mu_\rho=0$ throughout
\begin{equation}\tag{b),(j}
   \Psi_0=\frac{e^{\lambda-\mu}}{2}(2\lambda_\rho^2+\lambda_{\rho\rho}), \qquad
\Psi_2 = -\frac{e^{2(\lambda-\mu)}}{2}\Big[\lambda_\rho^2-\frac{\lambda_\rho}{\rho}\Big].
\end{equation}
There are three distinct limits to examine in this branch.
\begin{itemize}

\item{ $\rho\to0^+$.} Here $\ln\rho\to-\infty$, so
$\lambda_\rho\to+\infty$, since $e^{2(\lambda-\mu)}$ vanishes only
logarithmically, far too slowly to compensate the divergence due to the power-law, both scalars diverge, $\Psi_0\to-\infty$ and $\Psi_2\to+\infty$ exhibiting \textbf{a naked singularity} at the coordinate axis.

\item{ $\rho\to+\infty$.} Here $\lambda_\rho\to0$ and
$\lambda_{\rho\rho}\to0$, while $e^{\lambda-\mu}\to0$, both scalars vanish, $\Psi_0\to0$ and $\Psi_2\to0$, confirming this is an asymptotically conformally flat region.

\item{$\rho\to\rho_*$.} The bracket $C\ln\rho+\sqrt{C}B$ vanishes at $\rho_*=e^{-B/\sqrt{C}}$, a finite point within the domain $(0,\infty)$ regardless of the value of $B$. At this radius $\lambda\to+\infty$ and $e^{2\lambda}\to\infty$, giving another curvature singularity, occurring at a point unconnected to any coordinate degeneracy. We use the term ``singular radius" specifically for $\rho_*$ to emphasize this feature.

\end{itemize}

Based on the previous points, this branch is singular at both boundaries of the physical range as well as at one interior point, and regular only in the single asymptotic limit $\rho\to\infty$.

The $k<0$ branch behaves quite differently at both ends. Since $\ln\rho$ ranges over all of
$(-\infty,\infty)$ as $\rho$ ranges over $(0,\infty)$, the argument $\sqrt{|k|}\ln\rho+\delta$ of the cosine is unbounded in both directions, so it passes through infinitely many zeros as $\rho\to0^+$ and as $\rho\to+\infty$. This branch therefore develops infinite curvature singularities rather than a single isolated singular radius which is a qualitatively richer, and more pathological, singularity structure than the $k=0$ branch.

Next, in pursuit of an explicit relation of $\mu(\rho)$ we eliminate $V_\rho^2/e^{2\lambda}$ from \eqref{g+p} using \eqref{eq:anEV} as follows
\begin{equation}
    \mu_\rho=\rho\lambda_\rho^2-\rho\lambda_{\rho\rho}-\lambda_\rho.
\end{equation}
One can easily shows that $\lambda_\rho=-\frac{1}{\rho}\frac{\dot\varphi}{\varphi}$ and using $\ddot\varphi=k\varphi$ we obtain
\begin{equation}
\mu(\rho) = k\ln\rho+\mu_0.
\end{equation}
This result holds for any potential obtained above regardless of the sign of $k$, but since only the $k=0$ and $k<0$ branches of $\lambda(\rho)$ are physically admissible, only those two values of $k$ enter $\mu(\rho)$. For $k=0$ branch $\mu(\rho)=\text{const.}$ and decreases
logarithmically, $\mu=k\ln\rho+\mu_0$ with $k<0$.

The corresponding Killing tensor, for every one of these branches, is fixed by
\begin{equation}
\lambda_0=0, \qquad \lambda_1=\text{const}, \qquad \lambda_3=C_3\,e^{2(\mu-\lambda)}, \qquad
\lambda_2=-\lambda_1-\lambda_3,
\end{equation}
and therefore takes the following form denoting its reducibility to Killing tensors of equal (metric) or lower rank (Killing vectors)
\begin{equation}
K^{\mu\nu} = \lambda_1\,g^{\mu\nu} - 2\,\big(\lambda_3\,e^{-2(\mu-\lambda)}\big)\,\xi_z^\mu\xi_z^\nu
= \lambda_1\,g^{\mu\nu} - 2C_3\,\xi_z^\mu\xi_z^\nu.
\end{equation}

This decomposition holds for \emph{any} admissible $V(\rho)$ and, by
extension, any $\lambda(\rho)$ satisfying the field equations above, vacuum or electrovacuum, and regardless of which of the three branches is chosen. Notably, $\lambda_3(\rho)$ itself is not constant but an explicit function of $\rho$ through whichever background lives because $\lambda_3\propto e^{2(\mu-\lambda)}$ is exactly the solution of $\delta\lambda_3=-4\bar\alpha\lambda_3$, not a generic feature of an arbitrary $\lambda_3(\rho)$. The dependence of $\rho$ carried by $\lambda_3(\rho)$ and by the normalization of $\xi_z^\mu\xi_z^\nu$ is canceled. The Killing tensor in this branch is therefore reducible in vacuum and in electrovacuum alike, and encodes only the two linear invariants $E$ and $L_z=\xi_z\cdot p$ already guaranteed by the manifest isometries and not a fourth hidden constant of motion.

\subsection{Case $\lambda_0= \lambda_1+\lambda_2-\lambda_3=0$ and $\Phi_0-\bar\Phi_0\neq0$ }
We initiate our analysis by the annihilation of the imaginary part of $\kappa$ along with relation $\eqref{kappa}$ yielding the following relation
\begin{equation}
     \lambda_\rho = \frac{1}{2\rho} \quad \Rightarrow \quad e^{2\lambda} = \rho e^{\Sigma(z)}.
\end{equation}
The latter along with the annihilation of the real part of $\Phi_0$ do not prohibit the dependence of $z$ coordinate and reshapes the field equations and Maxwell equations as follows
\begin{equation}\label{eq:NP512i}
\lambda_{zz} = e^{-2\lambda}V_z^2
\end{equation}
\begin{equation}\label{eq:NP512ii}
\mu_\rho = \rho\Big(\frac{1}{4\rho^2} - \lambda^2_z \Big) -\rho e^{-2\lambda}V_z^2, 
\end{equation}
\begin{equation}\label{eq:NP512iii}
        \mu_z = \lambda_z,
\end{equation}
\begin{equation}\label{eq:NP512iv}
V_{zz} = 2\lambda_z V_z
.\end{equation}
Equivalently, \eqref{eq:NP512i} becomes
\begin{equation}
\lambda_{zz} = \frac{V_z(z)^2}{\rho\,e^{\Sigma(z)}}.
\end{equation}
The left-hand side, $\lambda_{zz}=\tfrac12\Sigma''(z)$, is manifestly independent of $\rho$, while the right-hand side carries an explicit, unremovable $1/\rho$ dependence unless $V_z\equiv0$. Consistency for all $\rho$ therefore yields 
\begin{equation}
\Sigma''(z)=0 \qquad\text{and}\qquad V_z(z)=0
\end{equation}
independently -- the second of which directly contradicts the non-triviality requirement $V_z\neq0$ demanded by the branch title. Hence no genuinely charged solution exists on this branch, thus $V=\text{const}$ and electrovacuum reduces to vacuum, and simultaneously fixes $\Sigma(z)=c_1z+c_2$ to be affine, recovering exactly the vacuum background $\lambda=\tfrac12\ln\rho+cz$ already identified. This is a considerably more direct route to the same conclusion than solving the general coupled system below, but we record that system for completeness, since it is what one would need without first invoking the alignment condition.

The Killing equations set up an additional constraint to this system 
\begin{equation}
    \delta(\lambda_1+\lambda_2-\lambda_3)=0 \Rightarrow (\pi+\alpha+\bar\alpha)\lambda_3=0.
\end{equation}
Considering that the annihilation of $\lambda_3$ gives the trivial case we choose $\pi+\alpha+\bar\alpha=0$ obtaining an extra constrain 
\begin{equation}
    (\mu-\lambda)_\rho = \frac{1}{2\rho} \quad\Rightarrow e^{2(\mu-\lambda)} = \rho e^{h(z)} .
\end{equation}
With these constraints we substitute into the field equations
\eqref{eq:NP512i}-\eqref{eq:NP512iii} and after simplification we get
\begin{align}
V_z^2 &= \frac{\rho}{2}\Sigma'' e^{\Sigma}, \label{eq:A}\\
V_z^2 &= e^{\Sigma}\Big[\frac{3}{4\rho}+\frac{\rho}{4}\Sigma'^2\Big], \label{eq:B}\\
h'(z) &= 0. \label{eq:C}
\end{align}
Equations \eqref{eq:A}--\eqref{eq:B} determine $V_\rho^2$ and $V_z^2$ individually; for a real
function $V(\rho,z)$ to exist, these must be consistent with \eqref{eq:C} through the identity
$(V_\rho V_z)^2=V_\rho^2V_z^2$. Imposing this yields
\begin{equation}
 \rho\Sigma(z)'' \;-\; \frac{\rho}{2}\Sigma(z)'^2 -\frac{3}{2\rho}=0.
\end{equation}
The last condition is inevitable, independent of any choice of $\Sigma(z)$. Even disregarding it, the second condition collapses to $V_z^2=0$, yielding\ $V=\mathrm{const}$. Thus, we conclude that the separable ansatz admits a trivial electrostatic potential and the only consistent solution is the trivial uncharged case $\Sigma,h=\mathrm{const}$, $V=\mathrm{const}$.

\subsubsection{Case $\lambda_3= \lambda_0+\lambda_1+\lambda_2=0$ and $\Phi_0 +\bar\Phi_0\neq0$}
Both remaining entries of the $q=+1$ table sit inside the $\Phi_0+\bar\Phi_0\neq0$ branch of the constraint \eqref{eq:qKM1}. Within $\Phi_0+\bar\Phi_0\neq0$, the two remaining cases are distinguished simply by whether $\lambda_3$ vanishes or not.

As before, annihilation of the imaginary part of $\Phi_0$ forces $V_z=0$ and hence $\lambda_z=0$, so $\xi_z=\partial_z$ is manifest and the geodesic motion is already completely separated by $\xi_t,\xi_\varphi,\xi_z$. Using $\lambda_3=0$ to eliminate the $(m,\bar m)$-sector in favor of $l,n,g$, the Killing tensor reduces to
\begin{equation}
K_{\mu\nu} = 2\lambda_0\,\rho^2e^{-2\lambda}\,(d\varphi)_\mu(d\varphi)_\nu - \lambda_2\,g_{\mu\nu},
\end{equation}
and raising indices, the $\rho,\lambda$-dependence cancels identically, leaving
\begin{equation}
K^{\mu\nu} = -\lambda_2\,g^{\mu\nu} + 2C_0\,\xi_\varphi^\mu\xi_\varphi^\nu.
\end{equation}
This reducible Killing tensor contributing nothing beyond $m^2$ and $L$, consistent with the separability already guaranteed by $\xi_t,\xi_\varphi,\xi_z$.

\subsubsection{Case $\lambda_0+\lambda_1+\lambda_2+\lambda_3=0,\ \lambda_3\neq0$ and $\Phi_0 +\bar\Phi_0\neq0$}

Here $\lambda_0\neq0$ survives, and the same $\lambda_z=0$ collapse applies (Section~\ref{Section6.1.1}). Integration of the Killing equations gives
\begin{equation}
\lambda_3=C_0\,e^{2(\mu-\lambda)}, \qquad \lambda_2=B_0-\lambda_3, \qquad \lambda_0=A_0\,\rho^2e^{-2\lambda}, \qquad \lambda_1=-\lambda_0-B_0,
\end{equation}
so that $\lambda_2+\lambda_3=B_0$ is constant, and the Killing tensor is
\begin{equation}
K^{\mu\nu} = -B_0\,g^{\mu\nu} + \beta\,\xi_t^\mu\xi_t^\nu + \gamma\,\xi_\varphi^\mu\xi_\varphi^\nu - 2C_0\,\xi_z^\mu\xi_z^\nu,
\end{equation}
again fully reducible with respect to $\{\partial_t,\partial_\varphi,\partial_z\}$ carrying no independent dynamical content.
Both cases therefore confirm the same conclusion even if $\lambda_3$ vanishes or not. The resulting Killing tensor is fully absorbed by the three Killing vector algebra that $\lambda_z=0$ guarantees.

\section{Summary \& Discussion}\label{Chapter 8}
The initial scope of this work was to confront Weyl's class as a whole, both in vacuum and electrovacuum, foreseeing that there would be a combination of coefficients of the generic asymptotically flat multipole expansion potential which, along with an irreducible Killing tensor, would be able to determine integrable spacetimes. In addition, it would be an analytical approach to clarify the non-integrability of the already known subclasses of Weyl's class. The results obtained resolve the question posed in a more structured way than a simple negative answer. To our knowledge, until now there has been no systematic classification of all canonical rank-2 Killing tensors admitted by the entire Weyl class. The result results in a criterion that decides, for any of these subclasses of Weyl's class, whether a genuine hidden symmetry is present. Also, we note that we adopt the notion of reducibility presented in \cite{dolan1989significance} where a Killing tensor is reducible if it can be written as a combination of the metric and symmetrized products of Killing vectors with constant coefficients and is genuinely admitted by the background.

\begin{center}
\begin{table}[H]
\caption{Killing tensors of Weyl's class}
\label{tab:KillingTensorSummary}
\footnotesize
\centering
\begin{tabular}{c | c | c}
Case & Constraint & Killing tensor $K^{\mu\nu}$ \\ \hline
$q=0,\ \kappa+\bar\kappa=0\neq\kappa-\bar\kappa$ & $\lambda_0=0$ & $\lambda_1\,g^{\mu\nu}$ \\
$q=0,\ \kappa-\bar\kappa=0,\ \text{Levi-Civita}$ & $\kappa-\bar\kappa=0$ & $B_0\,g^{\mu\nu}-2C_3\,\xi_z^\mu\xi_z^\nu$ \\
$q=-1,\ \kappa+\bar\kappa=0\neq\kappa-\bar\kappa$ & $\delta\lambda_0=2\bar\pi\lambda_0$ & $A\,\xi_{(t}^\mu\xi_{\phi)}^\nu+\lambda_1\,g^{\mu\nu}$ \\
$q=-1,\ \kappa-\bar\kappa=0,\ \text{Levi-Civita}$ & $\kappa-\bar\kappa=0$ & $A\big(\xi_t^\mu\xi_\phi^\nu+\xi_\phi^\mu\xi_t^\nu\big)+\lambda_1\,g^{\mu\nu}-2C_3\,\xi_z^\mu\xi_z^\nu$ \\
$q=+1,\ \kappa\neq0,\ \lambda_3=0$ & $\lambda_0+\lambda_1+\lambda_2=0$ & $N_1\xi_t^\mu\xi_t^\nu+N_2\xi_\varphi^\mu\xi_\varphi^\nu-\lambda_2\,g^{\mu\nu}$ \\
$q=+1,\ \kappa\neq0,\ \lambda_3\neq0$ & $C_3\,e^{2(\mu-\lambda)}=F(\rho)+G(z)$ & No generic potential \\ \hline
$q=0,\pm1,\ \lambda_0=\lambda_1+\lambda_2+\lambda_3=0$ & $\Phi_0+\bar\Phi_0\neq0$ & $\lambda_1\,g^{\mu\nu}-2C_3\,\xi_z^\mu\xi_z^\nu$ (Bronnikov) \\
$q=0,\pm1,\ \lambda_0=\lambda_1+\lambda_2-\lambda_3=0$ & $\Phi_0-\bar\Phi_0\neq0$ &  Vacuum Levi-Civita \\ 
$q=+1,\ \lambda_3=0$ & $\lambda_0+\lambda_1+\lambda_2=0$ & $-\lambda_2\,g^{\mu\nu}+2C_0\,\xi_\varphi^\mu\xi_\varphi^\nu$ \\
$q=+1,\ \lambda_3\neq0$ & $\lambda_0+\lambda_1+\lambda_2+\lambda_3=0$ & $-B_0\,g^{\mu\nu}+\beta\,\xi_t^\mu\xi_t^\nu+\gamma\,\xi_\varphi^\mu\xi_\varphi^\nu-2C_0\,\xi_z^\mu\xi_z^\nu$ \\
\end{tabular}
\end{table}
\end{center}
\normalsize
The table above summarizes all cases examined in this work, organized into the vacuum and electrovacuum limit.
\subsection*{Vacuum}
In the vacuum limit the Levi-Civita solutions were not presented in detail because we set the target to cases which had the potential to admit irreducible Killing tensors. Indeed, the existence of $\partial_z$ as a third Killing vector combined with the general solution of the Killing equations produces unavoidable reducibility of the Killing tensor within Weyl's class. Solving the Killing equations for the canonical tetrads, we obtain a Killing tensor of the form $K^{\mu\nu}=\alpha g^{\mu\nu}+\beta\xi_t^\mu\xi_t^\nu+\gamma\xi_\varphi^\mu\xi_\varphi^\nu+\delta\xi_z^\mu\xi_z^\nu$ for all $\rho$. Consequently $Q=K^{\mu\nu}p_\mu p_\nu$ reduces to a linear combination of quantities already conserved by the three Killing vectors and the mass shell, carrying no independent information. In short, wherever a Weyl-class background is exactly $z$-independent, the Killing tensor is reducible, and no genuine hidden symmetry can arise in this subclass.
The Killing tensor with $q=+1$ (Case $\lambda_3=0$) was the most prominent candidate for irreducibility, since equation \eqref{eq:caseAsol} holds on \emph{any} static, axisymmetric vacuum Weyl background, for every harmonic $\lambda$, including any asymptotically flat multipole solution, Schwarzschild, Zipoy-Voorhees, or Levi-Civita, with $N_1,N_2,\lambda_2$ constants. No specific form of $\lambda$ was assumed at any stage, only the general expressions for $\pi,\kappa$ were used, and the companion potential $\mu$ never appears explicitly. However, a direct computation of its contravariant components reveals that the corresponding Killing tensor is reducible. Thus, this universal branch of the $q=+1$ canonical form does not encode a genuine hidden symmetry on any background. This provides a sharp and general result: an entire family of apparent Killing-tensor candidates can be excluded from consideration as sources of hidden symmetries.

The other branch of $q=+1$ with $\lambda_3\neq0$ is exactly what allows $K^{\rho\rho}\neq K^{zz}$, and no reducible combination built from $g^{\mu\nu}$ and $\xi_t,\xi_\phi$ alone can produce that feature. So whenever the case $\lambda_3\neq0$ exists, it is automatically independent of $E,L$. This case, though, is governed by a strict separability condition on the conformal factor of the meridional plane namely, $e^{2(\mu-\lambda)}$. This condition is equivalent to the classical Liouville criterion for a 2-dimensional metric, which places the whole discussion on track in which a Killing tensor of this type furnishes genuine new integrability precisely when the source can be represented by a single Weyl rod occupying one confocal coordinate patch, whether that rod is finite (Schwarzschild) or infinite (Levi-Civita). Every configuration built from more than one rod, or from a non-uniform density along a single rod such as multipole sources, fails this criterion and correspondingly lacks an additional integral of motion, regardless of the Petrov type of the solution in question.

Built in coordinates adapted to Schwarzschild's own confocal structure, it reproduces Carter's constant $Q=p_\theta^2+L^2/\sin^2\theta$ exactly, and cannot be written as any constant-coefficient combination of the metric tensor and Killing vectors. This is structurally the same as the mechanism established for the Levi-Civita family where $\partial_z$ is an extra Killing vector invisible to the generic $q=0,\pm1$ classification, present because Levi-Civita happens to be more symmetric than a generic Weyl background. Here, $\xi_x,\xi_y$ are extra Killing vectors invisible to the axisymmetric formalism, present because Schwarzschild happens to be spherically symmetric. In both cases, apparent irreducibility relative to the algebra $\{\xi_t,\xi_\varphi\}$ is an artifact of testing against a smaller symmetry group than the background actually possesses. Carter's constant is considered new information and it is precisely what is required to separate the Hamilton-Jacobi equation when the generic axisymmetric formalism is considered. In other words, its integrability is real and complete, but it is explained entirely by non-hidden symmetry.
Along these lines, we present one clean statement: within Weyl's class, there is no such a background able to admit a Killing tensor that is irreducible. Schwarzschild remains the only background examined whose Killing tensor achieves \emph{dynamical} novelty relative to the generic axisymmetric subalgebra $\{\xi_t,\xi_\varphi\}$ which is solely related to isometries than hidden symmetries. The $\lambda_3=0$ branch is reducible on every background. Case $\lambda_3\neq0$ branch of $q=+1$, is also reducible discribed by the separation a criterion which is not satisfied by any background as a whole.

\subsection*{Electrovacuum}
Electrovacuum limit considered as a wider playground for metric potentials hoping for the existence of spactimes with hidden symmetries. As a matter of fact there is not such a spacetime however, within the cases examined an interesting solution emerged containing naked singularities at $\rho=\rho_*$ and at the axis of symmetry.This solution was first obtaine by Bronnikov in  \cite{bronnikov1979inverted} (see also \cite{bronnikov2020cylindrical}) while the author investigates static, axially symmetric electrovacuum spacetimes by adopting a more general metric ansatz for the $g_{\phi\phi}$ component. 

Remarkably, this broader ansatz leads to precisely the same spacetimes that we obtain in Sec.~\ref{Section6.1.1}. All of our cases can be identified with solutions within the classification of \cite{bronnikov1979inverted} however, the branch corresponding to $k>0$ is not physically admissible within our setup. This difference stems from the fact that Weyl's class constitutes a restricted subclass of the more general metric considered by Bronnikov, allowing for an additional branch that is excluded by the Weyl form of the metric. More specifically, in Bronnikov's parametrization, the $g_{\phi\phi}$ component is given by the square of a function of the non-ignorable coordinate, rather than by the Weyl form $g_{\phi\phi}=-\rho^2 e^{-2\lambda}$. The latter imposes an additional restriction on the metric and consequently on the corresponding field equations. This observation is further supported, from a different perspective\textcolor{blue}{,} by the recent analysis of \cite{nasereldin2026static}, where it is emphasized that the Weyl form of the metric imposes an additional constraint\textcolor{blue}{:} taking $g_{\phi\phi}=-\rho^2 e^{-2\lambda}$ is incompatible with a non-vanishing cosmological constant. This remark also emerges in our case\textcolor{blue}{, in equations} \eqref{lambda}.

\bibliographystyle{unsrt}
\bibliography{sn-bibliography}

\end{document}